\documentclass[a4paper,11pt]{article}

\usepackage{amsmath}
\usepackage{amsfonts}
\usepackage{amssymb}
\usepackage{caption}
\usepackage{graphicx}
\usepackage{graphics}
\usepackage{color}
\usepackage{ushort}
\usepackage{makeidx}
\usepackage{float}
\usepackage{appendix}
\usepackage{accents}
\usepackage{framed}
\usepackage{versions}
\usepackage{xcolor}
\usepackage{mathbbol}
\usepackage{blindtext}
\usepackage{emptypage}
\usepackage{verbatim}
\usepackage{enumitem}
\usepackage{array}
\usepackage{multirow}

\def\hmath$#1${\texorpdfstring{{\rmfamily\textit{#1}}}{#1}}

\usepackage[bookmarks]{hyperref}
\hypersetup{
    colorlinks=true,   	
    linkcolor=red,      
    citecolor = [rgb]{0 0.7 0},   	
    filecolor=magenta, 	
    urlcolor=blue
}

\newcommand{\Xp}{\mbox{\boldmath $X$}}

\newcommand{\Yp}{\mbox{\boldmath $Y$}}

\newcommand{\Zp}{\mbox{\boldmath $Z$}}

\newcommand{\RL}{{\mathbb R}}
\newcommand{\NN}{{\mathbb N}}

\newcommand{\IND}{{\mathbb I}}
\newcommand{\BBP}{{\mathbb P}}

\newcommand{\BBE}{{\mathbb E}}

\newcommand{\VAR}{\mbox{\rm Var}}

\def\ba{\begin{align}}
\def\ea{\end{align}}
\def\ban{\begin{align*}}
\def\ean{\end{align*}}

\def\be{\begin{eqnarray}}
\def\ee{\end{eqnarray}}
\def\ben{\begin{eqnarray*}}
\def\een{\end{eqnarray*}}

\def\bqq{\begin{equation}}
\def\eqq{\end{equation}}
\def\bqqn{\begin{equation*}}
\def\eqqn{\end{equation*}}

\def\elabel#1{\label{e:#1}}

\def\sq{$\Box$}

\def\qed{\ifmmode\sq\else{\unskip\nobreak\hfil
\penalty50\hskip1em\null\nobreak\hfil\sq
\parfillskip=0pt\finalhyphendemerits=0\endgraf}\fi\par\medbreak}

\newsavebox{\junk}
\savebox{\junk}[1.6mm]{\hbox{$|\!|\!|$}}

\def\limsup{\mathop{\rm lim\ sup}}

\def\til={{\widetilde =}}

\def\clX{{\cal X}}
\def\clY{{\cal Y}}

 \def\eq#1/{(\ref{#1})}

\newtheorem{theorem}{Theorem}[section]

\newtheorem{proposition}[theorem]{Proposition}
\newtheorem{lemma}[theorem]{Lemma}
\newtheorem{definition}[theorem]{Definition}

\def\eq#1/{(\ref{e:#1})}

\newcommand{\beqn}[1]{\notes{#1}%
\begin{eqnarray} \elabel{#1}}

\newcommand{\eeqn}{\end{eqnarray} } 

\newcommand{\beq}[1]{\notes{#1}%
\begin{equation}\elabel{#1}}

\newcommand{\eeq}{\end{equation}} 

\def\bdes{\begin{description}}
\def\edes{\end{description}}

\DeclareMathOperator*{\Bcup}{\text{\raisebox{0.25ex}%
	{\scalebox{0.8}{$\bigcup$}}}}

\def\notes#1{}

\definecolor{mag}{rgb}{0.7,0,0.3}
\definecolor{dgreen}{rgb}{0.1,0.5,0.1}
\definecolor{dred}{rgb}{.8,0,0}
\definecolor{gray}{rgb}{.8,.8,.8}
\definecolor{brown}{rgb}{0.6451,0.3706,0.1745}

\newcommand{\E}{\mathbb{E}}
\newcommand{\PP}{\mathbb{P}}
\newcommand{\Pbig}[1]{%
    \PP\bigl[ #1\bigr]%
}
\newcommand{\seq}[1]{%
    \{ #1\}%
}

\newenvironment{proof}{\paragraph{Proof. }}{\hfill$\square$}

\begin{document}

\title{\vspace{-1.5cm}%
Fundamental Limits of Lossless Data Compression\\
with Side Information}

\author
{
	Lampros Gavalakis
    \thanks{Department of Engineering,
	University of Cambridge,
        Trumpington Street,
	Cambridge CB2 1PZ, U.K.
                Email: \texttt{\href{mailto:lg560@cam.ac.uk}%
			{lg560@cam.ac.uk}}.
	L.G.\ was supported in part by EPSRC grant number RG94782.
        }
\and
        Ioannis Kontoyiannis 
    \thanks{Statistical Laboratory, DPMMS,
	University of Cambridge,
	Centre for Mathematical Sciences,
        Wilberforce Road,
	Cambridge CB3 0WB, U.K.
                Email: \texttt{\href{mailto:yiannis@maths.cam.ac.uk}%
			{yiannis@maths.cam.ac.uk}}.
		Web: \texttt{\url{http://www.dpmms.cam.ac.uk/person/ik355}}.
	I.K.\ was supported in part by the Hellenic Foundation for Research 
	and Innovation (H.F.R.I.) under the ``First Call for H.F.R.I. Research 
	Projects to support Faculty members and Researchers and the 
	procurement of high-cost research equipment grant,'' project 
	number 1034.
        }
}

\date{\today}

\maketitle

\begin{abstract}
The problem of lossless data compression 
with side information available to both the
encoder and the decoder is considered.
The finite-blocklength fundamental limits 
of the best achievable performance 
are defined, in two different versions
of the problem: {\em Reference-based compression},
when a single side information string
is used repeatedly in compressing different
source messages, and {\em pair-based compression},
where a different side information string
is used for each source message. 
General achievability and converse theorems
are established for arbitrary source-side
information pairs.
Nonasymptotic
normal approximation expansions are
proved for the optimal rate in both the reference-based
and pair-based settings,
for memoryless sources.
These are stated in terms of explicit,
finite-blocklength bounds, that are tight
up to third-order terms.
Extensions that go significantly beyond the class of memoryless
sources are obtained.
The relevant source dispersion is identified
and its relationship with the conditional 
varentropy rate is established.
Interestingly, the dispersion is different
in reference-based and pair-based compression,
and it is proved that the reference-based 
dispersion is in general smaller.
\end{abstract}



\newpage

\section{Introduction}

It has long been recognised in information theory
\cite{slepianwolf:73,cover:book2} that the presence 
of correlated side information can
dramatically improve compression performance.
Moreover, in many applications useful side
information is naturally present. 

{\em Reference-based compression. } A particularly
important and timely application of compression with 
side information is to the problem of storing the vast 
amounts of genomic data currently being generated 
by modern DNA sequencing technology \cite{pavlichin:18}.
In a typical scenario, the genome~$X$ of 
a new individual that needs to be stored
is compressed using a reference genome~$Y$ 
as side information. Since most of the time
$X$ will only be a minor variation of $Y$,
the potential compression gains are large.
An important aspect of this scenario is that the
same side information -- in this case the
reference genome $Y$ -- is used 
in the compression of many
new sequences~$X^{(1)},X^{(2)},\ldots$. We call this 
the {\em reference-based}
version of the compression problem.

{\em Pair-based compression. } 
Another important application of compression with
side information is to 
the problem of file synchronisation \cite{rsync},
where updated computer
files need to be stored along with their earlier versions, 
and the related problem of 
software updates \cite{suel:02},
where remote users need to be provided with newer versions
of possibly large software suites. Unlike 
genomic compression,
in these cases a different side
information sequence $Y$ (namely, the older version of the 
specific file or of the particular software) 
is used every time a new piece of
data~$X$ is compressed. We refer to this as the
{\em pair-based} version of the compression problem,
since each time a different $(X,Y)$ pair is considered.
Other areas where pair-based
compression is used include, among many others,
the compression of noisy versions
of images \cite{pradhan:01}, and the 
compression of future video frames given 
earlier ones \cite{aaron:02}.

In addition to those appearing in work already 
mentioned above, a number practical algorithms 
for compression with side information
have been developed over the past 25 years. 
The following are a some representative 
examples.
The most common approach is based on 
generalisations of the 
Lempel-Ziv compression methods
\cite{subrahmanya:95,uyematsu:03,tock:05,jacob:08,jain-bansal}.
Information-theoretic treatments of problems
related to DNA compression with side information
have been given, e.g., in \cite{yang:01,fritz:11,chen:14}.
A grammar-based algorithm was presented in~\cite{stites:00},
turbo codes were used in~\cite{aaron:02b},
and a generalisation of the context-tree weighting
algorithm was developed in~\cite{C-K-Verdu:06}.

In this work
we describe and evaluate
the fundamental limits of the best achievable
compression performance, when side information
is available both at the encoder and the decoder.
We derive tight, nonasymptotic
bounds on the optimum rate, 
we determine the source dispersion in both
the reference-based and the pair-based
cases, and we examine the difference between
them.

\subsection{Outline of main results}

In Section~\ref{descriptionofanoptimal} we
give precise definitions for the finite-blocklength
fundamental limits of reference-based and pair-based 
compression. 
We identify the theoretically optimal 
one-to-one compressor in each case, 
for arbitrary source-side information pairs
$(\Xp,\Yp)$ $=\{(X_n,Y_n)\;;\;n\geq 1\}$.
Moreover,
in Theorem~\ref{prefixpenaltytheorem}
we show that,
for any blocklength $n$, requiring the compressor
to be prefix-free imposes a penalty of no more
than $1/n$~bits per symbol on the optimal rate.

In Section~\ref{arbitrarysources} we state and prove
four general, single-shot, achievability and converse results,
for the compression of arbitrary sources with arbitrarily
distributed side information. 
Theorem~\ref{t:achieve} generalises
a corresponding result established in~\cite{kontoyiannis-verdu:14}
without side information.
Theorem~\ref{lemmalength},
one of the main tools we use later
to derive
the normal approximation results,
gives new, tight
achievability and converse bounds,
based on a counting argument. 

Sections~\ref{memorylessnormalsection}
and~\ref{annealedsectionmemoryless} contain
our main results, giving nonasymptotic, 
normal-approximation expansions to the optimal 
reference-based rate and the optimal pair-based rate.
These expansions give finite-$n$ upper and lower bounds that
are tight up to third-order terms.

For the sake of clarity,
we first describe the pair-based results
of Section~\ref{annealedsectionmemoryless}.
The {\em conditional entropy rate} 
of a source-side information pair
$(\Xp,\Yp)$ $=\{(X_n,Y_n)\;;\;n\geq 1\}$ is,
$$H(\Xp|\Yp)=\lim_{n\to\infty}\frac{1}{n}H(X_1^n|Y_1^n),\
\qquad\mbox{bits/symbol},$$ 
whenever the limit exists,
where 
$X_1^n=(X_1,X_2,\ldots,X_n)$,
$Y_1^n=(Y_1,Y_2,\ldots,Y_n)$,
and $H(X_1^n|Y_1^n)$ denotes the conditional entropy
of $X_1^n$ given $Y_1^n$.
Similarly,
the {\em conditional varentropy rate} \cite{nomura:11} is,
$$
\sigma^2(\Xp|\Yp)=\lim_{n\to\infty}\frac{1}{n}
\VAR\big(-\log P(X_1^n|Y_1^n)\big),$$
whenever the limit exists,
where $\log=\log_2$.
This generalises the {\em minimal coding variance}
of \cite{kontoyiannis-97};
precise definitions will be given in the
following sections.

Let $R^*(n,\epsilon)$ be the best pair-based compression
rate that can be achieved
at blocklength $n$,
with excess rate probability no greater than
$\epsilon$.
For a memoryless source-side information pair $(\Xp,\Yp)$,
in Theorems~\ref{IIDAnnealedAchievability} and~\ref{IIDAnnealedConverse}
we show that
there are finite constants $C,C'>0$ such that,
for all $n$ greater than some $n_0$, 
\begin{equation}
\label{introannealedresult}
-\frac{C'}{n} \leq R^*(n,\epsilon) 
- \Bigl[H(X|Y) + \frac{1}{\sqrt{n}}\sigma(X|Y) Q^{-1}(\epsilon) 
-\frac{\log{n}}{2n}\Bigr] \leq \frac{C}{n}.
\end{equation}
Moreover, 
explicit expressions are obtained for $n_0,C,C'$. 
Here 
$Q$ denotes the standard Gaussian tail function,
$Q(z)=1-\Phi(z)$, $z\in\RL$;
for memoryless sources,
the conditional entropy rate
reduces to $H(X|Y)=H(X_1|Y_1)$, 
and the conditional
varentropy rate becomes 
$\sigma^2(X|Y)= \VAR\bigl(-\log{P(X_1|Y_1)}\bigr).$

The bounds in~(\ref{introannealedresult}) generalise
the corresponding no-side-information
results in Theorems~17 and~18 of \cite{kontoyiannis-verdu:14};
see also the discussion in Section~\ref{s:related} for
a description of the natural connection with the Slepian-Wolf problem.
Our proofs rely on the general coding theorems 
of Section~\ref{descriptionofanoptimal}
combined with appropriate versions of the classical
Berry-Ess\'{e}en bound. An important difference 
with~\cite{kontoyiannis-verdu:14}
is that the approximation used in the proof of the upper 
bound in \cite[Theorem~17]{kontoyiannis-verdu:14}
does not admit a natural analog in the case of compression
with side information. Instead, we use the tight approximation 
to the description lengths of the optimal compressor
given in Theorem~\ref{lemmalength}.
This is a new result, 
in fact an improvement of Theorem~\ref{t:achieve},
and it does not have a known counterpart in the
no-side-information setting.

Results analogous to~(\ref{introannealedresult}) are
also established 
in a slightly weaker form
for the case of Markov sources
in Theorem~\ref{Markovboth}, which is the main content 
of Section~\ref{normalmarkovsection}.

In Section~\ref{memorylessnormalsection} we consider
the corresponding problem in the reference-based setting.
Given an arbitrary, fixed side information string $y_1^n$,
let $R^*(n,\epsilon|y_1^n)$ denote the best pair-based
compression rate that can be achieved at blocklength $n$,
conditional on $Y_1^n=y_1^n$,
with excess-rate probability no greater than~$\epsilon$.
Suppose that the distribution of $\Yp$ is arbitrary,
and the source $\Xp$ is conditionally i.i.d.\ 
(independent and identically distributed) given $\Yp$.

In Theorems~\ref{quenchedconverse}
and~\ref{quenchedupper} we prove reference-based
analogs of the bounds in~(\ref{introannealedresult}):
There are finite constants $C(y_1^n), C'(y_1^n)>0$
such that, for all $n$ greater than some $n_1(y_1^n)$,
we have,
\be
-\frac{C'(y_1^n)}{n}
\leq
 	R^*(n,\epsilon|y_1^n) - \Bigl[H_n(X|y_1^n)
	+ \frac{1}{\sqrt{n}}\sigma_n(y_1^n)
	Q^{-1}(\epsilon) 
	-\frac{\log{n}}{2n}\Bigr]
\leq 
	\frac{C(y_1^n)}{n},
	\label{introfixedrnresult}
\ee
where now the first-order rate is given by,
$$
H_n(X|y_1^n)=
\frac{1}{n}\sum_{i=1}^n{H(X|Y=y_i)},$$
and the variance $\sigma_n^2(y_1^n)$ is,
$$
	\sigma_n^2(y_1^n)
=\frac{1}{n}\sum_{i=1}^n  \VAR\big(-\log P(X|y_i)\big|Y=y_i\big).$$
Once again, explicit expressions are obtained
for $n_1(y_1^n),C(y_1^n)$ and $C'(y_1^n)$.
A numerical example illustrating the accuracy of
the normal approximation in~(\ref{introfixedrnresult})
is shown in Figure~\ref{f:rates}.

\begin{figure}[t!]
\vspace*{-0.15in}
\centerline{\includegraphics[width=4.0in]{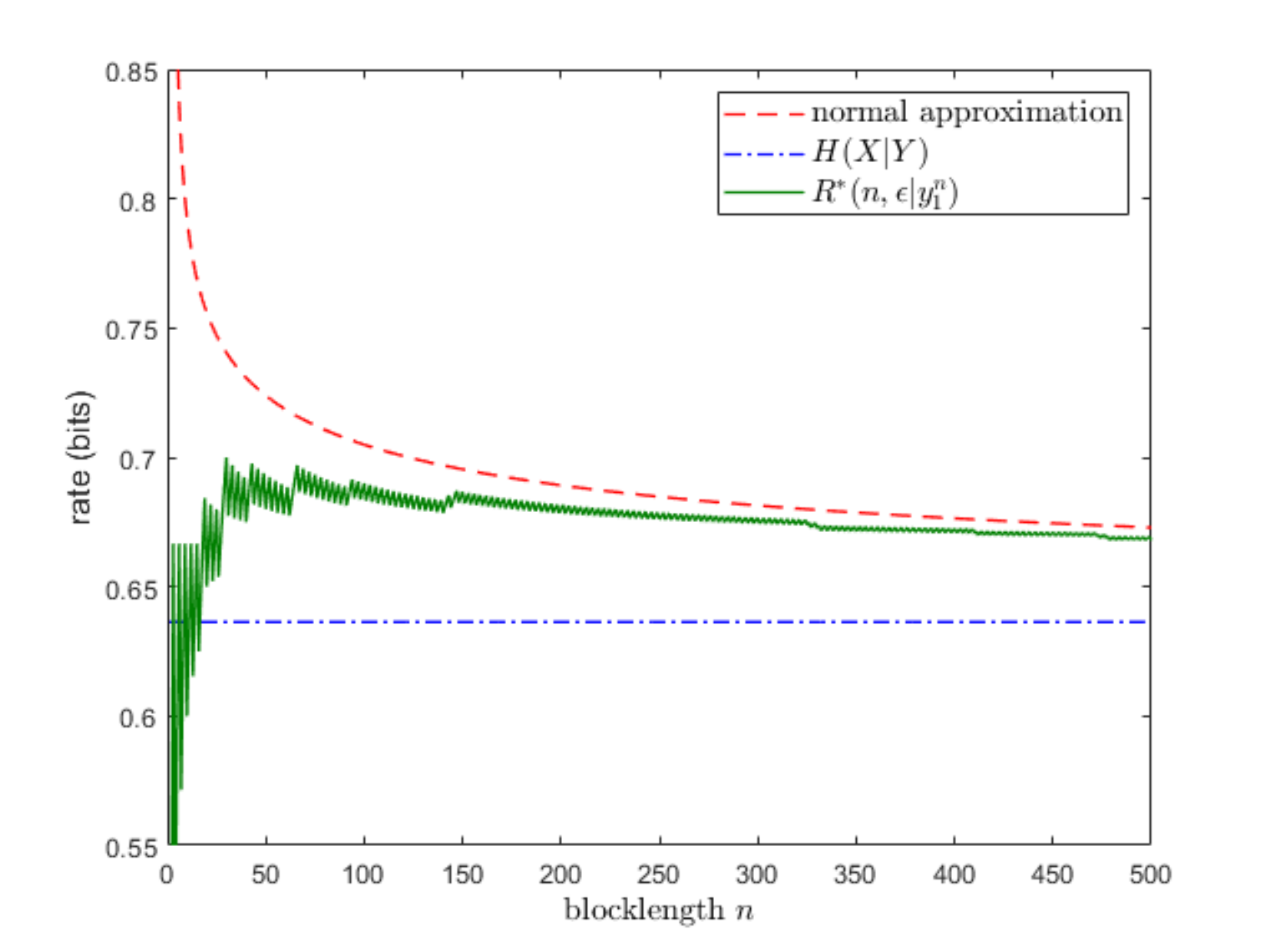}}
\vspace*{-0.05in}
\caption{Normal approximation to the 
reference-based optimal rate $R^*(n,\epsilon|y_1^n)$ 
for a memoryless side information process $\{Y_n\}$ with
Bern(1/3) distribution. The source
$\{X_n\}$ has
$X|Y=0\sim\mbox{Bern}(0.1)$ and
$X|Y=1\sim\mbox{Bern}(0.6)$. The conditional
entropy rate $H(X|Y)\approx 0.636$, whereas the
entropy rate of the source is $H(X)\approx0.837$.
The side information sequence is taken to be
$y_1^n=001001001\cdots$.
The graph shows $R^*(n,\epsilon|y_1^n)$ itself,
with $\epsilon=0.1$,
for blocklengths $1\leq n\leq 500$, together with the
normal approximation to 
$R^*(n,\epsilon|y_1^n)$ given by the three terms 
in square brackets
in~(\ref{introfixedrnresult}).}
\label{f:rates}
\end{figure}

Note that there is an elegant analogy between
the bounds in~(\ref{introannealedresult})
and~(\ref{introfixedrnresult}).
Indeed, there is further asymptotic
solidarity in the normal approximation
of the two cases. If $\Yp$ is ergodic,
then for a random side information string
$Y_1^n$ we have that, with probability~1,
$$H_n(X|Y_1^n)\to H(X|Y),
\qquad\mbox{as}\; n\to\infty.$$
But the corresponding variances
are different: With probability~1,
$$\sigma_n^2(Y_1^n)\to
\E\Big[\VAR\big(-\log P(X|Y)|Y\big)\Big],
\qquad\mbox{as}\; n\to\infty,$$
which is shown in Proposition~\ref{varentropyprop}~$(i)$
to be strictly smaller than $\sigma^2(X|Y)$ in general.
[The variance above is in terms of the conditional
distribution $P(X|Y)$, while the expectation
is with respect to~$Y$ only.]
This admits the intuitively satisfying interpretation
that, in reference-based compression, where a single side 
information string is used to compress multiple 
source messages, the optimal rate has smaller
variability.

\subsection{Related work}
\label{s:related}

The finer study of the optimal rate in source coding 
(without side information)
originated in Strassen's pioneering work~\cite{strassen:64b},
followed more than three decades later 
by~\cite{kontoyiannis-97} and more recently 
by~\cite{kontoyiannis-verdu:14}.
In addition to the works already described, 
we also mention that third-order normal approximation 
results in universal 
source coding were obtained in~\cite{irikosut:15}.

The dichotomy observed in the dispersion between
the reference-based and the pair-based settings
is in some ways analogous to the difference observed
in channel coding \cite{PPV:10}, between
the conditional information variance $V(P,W)$
and the unconditional variance $U(P,W)$,
in the context of
constant composition codes. 
In more recent work
\cite{sakai-arxiv:19}, a similar 
phenomenon was observed in source 
coding with side information, when examining 
the dispersion under maximum and average
error criteria.

The most direct relationship of the present 
development with current and recent work is 
in connection with the Slepian-Wolf (SW) problem~\cite{slepianwolf:73}.
A refined information spectrum-type bound for the SW converse
was derived by Miyake and Kanaya~\cite{miyake:95}.
Tan and Kosut~\cite{tan:12} give a second-order multidimensional 
normal approximation to the SW region 
for memoryless sources, and they show that,
up to terms of order $(\log n)/n$, achievable rates are
the same as if the side information
were known perfectly at the decoder.
Nomura and Han~\cite{nomura:14} derive the
second-order SW region for general sources.
Recently, Chen et al.~\cite{kostina:19} refined 
the results of~\cite{tan:12}
by establishing inner and outer asymptotic bounds 
for the SW region,
which are tight up to and including third-order terms.
Since, by definition, any SW code is also a pair-based code
for our setting, the achievability result
from \cite{kostina:19} implies a slightly weaker 
form of our Theorem~\ref{IIDAnnealedAchievability},
with an asymptotic $O(1/n)$ term in place of 
the explicit $C/n$ in~(\ref{eq:iidachieve}).
It is interesting to know that this high level
of accuracy in bounding above $R^*(n,\epsilon)$
can be derived both by random coding as in \cite{kostina:19}
and by deterministic methods as in Theorem~\ref{IIDAnnealedAchievability}.
The sharpest known SW converse is
obtained in \cite{jose:18} via 
a linear programming argument.

\section{The Optimal Compressor and Fundamental Limits}
\label{descriptionofanoptimal}

Let $(\Xp,\Yp)=\{(X_n,Y_n)\;;\;n\geq 1\}$ be a
{\em source-side information pair}, that is, a pair of arbitrary,
jointly distributed sources with finite alphabets
$\clX,\clY$, respectively, where $\Xp$ is to be compressed
and $\Yp$ is the side information process.
Given a source string
$x_1^n=(x_1,x_2,\ldots,x_n)$ and assuming
$y_1^n=(y_1,y_2,\ldots,y_n)$ is available to both the encoder
and decoder, a {\em fixed-to-variable one-to-one compressor 
with side information}, of blocklength $n$,
is a collection of functions $f_n$, where each
$f_n(x_1^n|y_1^n)$ takes a value in the set of all finite-length
binary strings,
$$\{0,1\}^*=\Bcup_{k=0}^\infty\{0,1\}^k=\{\emptyset,0,1,00,01,000,\ldots\},$$
with the convention that $\{0,1\}^0=\{\emptyset\}$ consists 
of just the empty string
$\emptyset$ of length zero. For each $y_1^n\in\clY^n$,
$f_n(\cdot|y_1^n)$ is assumed to be an injective function
from $\clX^n$ to $\{0,1\}^*$,
so that the compressed string $f_n(x_1^n|y_1^n)$ is always 
uniquely and correctly decodable. 
The associated description lengths of $\{f_n\}$ 
are,
$$\ell(f_n(x_1^n|y_1^n))
=\mbox{length of}\;f_n(x_1^n|y_1^n),\qquad\mbox{bits},$$
where, throughout, $\ell(s)$ denotes the length, in bits, of 
a binary string $s$.
For $\mbox{$1\leq i\leq j\leq\infty$}$, we use the shorthand 
notation $z_i^j$ for the string $(z_i,z_{i+1},\ldots,z_j)$,
and similarly $Z_i^j$ for the corresponding collection 
of random variables $Z_i^j=(Z_i,Z_{i+1},\ldots,Z_j)$.

The following fundamental limits describe the best achievable
performance among one-to-one compressors with side information,
in both the reference-based and the pair-based
versions of the problem, as described in the Introduction.

\begin{definition}[Reference-based optimal compression rate
\boldmath{$R^*(n,\epsilon|y_1^n)$}]
\label{quenchedrndef}
$\!\!\!$ For any blocklength $n$,
any fixed side information string $y_{1}^{n} \in \mathcal{Y}^n$,
and any $\epsilon\in[0,1)$,
we let $R^*(n,\epsilon | y_{1}^{n})$ denote the smallest 
compression rate that can be achieved with 
excess-rate probability no larger than $\epsilon$.
Formally, $R^*(n,\epsilon | y_{1}^{n})$ is the infimum
among all $R>0$ such that,
\ben
\min_{f_{n}(\cdot|y_{1}^{n})}{\mathbb{P}
\left[ \ell(f_{n}(X_{1}^{n}|y_{1}^{n})) > nR
|Y_1^n=y_1^n\right] \leq \epsilon},
\een
where the minimum is over all one-to-one compressors 
$f_n(\cdot|y_1^n):\clX^n\to\{0,1\}^*$.
\end{definition}

\noindent
Throughout, we write $\PP$ for the underlying
probability measure with respect to which all relevant probabilities
are computed, and the expectation operator $\E$ denotes 
integration with respect to~$\PP$.

\begin{definition}
[Pair-based optimal compression rate \boldmath{$R^*(n,\epsilon)$}]
\label{rnstardef}
$\;$ For any blocklength $n$
and any $\epsilon\in[0,1)$, we let $R^*(n,\epsilon)$
denote the smallest
compression rate that can be achieved with 
excess-rate probability, over both $X_1^n$ and $Y_1^n$,
no larger than $\epsilon$.
Formally, $R^*(n,\epsilon)$ is the infimum
among all $R>0$ such that,
\ben
\min_{f_{n}}{\mathbb{P}
\left[ \ell(f_{n}(X_{1}^{n}|Y_{1}^{n})) > nR\right]\leq\epsilon},
\een
where the minimum is over all one-to-one 
compressors $f_n$ with side information.
\end{definition}

\begin{definition}
[Reference-based excess-rate probability
\boldmath{$\epsilon^*(n,k|y_1^n)$}]
\label{condepsilonstardef}
$\;$ For any blocklength $n$, 
any fixed side information string $y_1^n\in\clY^n$, and any $k\geq 1$,
let $\epsilon^*(n,k|y_1^n)$ be the best achievable excess-rate probability
with rate $R=k/n$,
\ben
\epsilon^*(n,k|y_1^n) = 
\min_{f_n(\cdot|y_1^n)}\Pbig{\ell(f_n(X_1^n|y_1^n)) \geq k|Y_1^n=y_1^n},
\een
where the minimum is over all one-to-one compressors 
$f_n(\cdot|y_1^n):\clX^n\to\{0,1\}^*$.
\end{definition}

\begin{definition}
[Pair-based excess-rate probability
\boldmath{$\epsilon^*(n,k)$}]
\label{epsilonstardef}
For any blocklength $n$ and any $k\geq 1$,
let $\epsilon^*(n,k)$ be the best achievable excess-rate probability
with rate $R=k/n$,
\ben
\epsilon^*(n,k) = \min_{f_n}\Pbig{\ell(f_n(X_1^n|Y_1^n)) \geq k},
\een
where the minimum is over all one-to-one 
compressors $f_n$ with side information.
\end{definition}

\noindent
Before establishing detailed results on these fundamental limits
in the following sections, some remarks are in order.

\medskip

\noindent
{\bf The optimal compressor \boldmath{$f_n^*$}. }
It is easy to see from 
Definitions~\ref{quenchedrndef}--\ref{epsilonstardef}
that, in all four cases, the minimum is achieved by the same
simple compressor $f_n^*$: For each side information string $y_1^n$,
$f_n^*(\cdot|y_1^n)$ is the optimal compressor for 
the distribution $\BBP(X_1^n=\cdot|Y_1^n=y_1^n)$, namely, 
the compressor that orders the strings $x_1^n$
in order of decreasing probability $\BBP(X_1^n=x_1^n|Y_1^n=y_1^n)$, 
and assigns them codewords from $\{0,1\}^*$ in lexicographic order;
cf.~Property~1 in \cite{kontoyiannis-verdu:14}. 

\medskip

\noindent
{\bf Equivalence of minimal rate and excess-rate probability. }
The following relationships are straightforward from the definitions: 
For any $n,k\geq 1$,
any $\epsilon\in[0,1)$, and all $y_1^n\in\clY^n$:
\be
R^*(n,\epsilon|y_1^n) 
&=&
	\frac{k}{n} \qquad \text{iff}\qquad  \epsilon^*(n,k+1|y_1^n) 
	\leq \epsilon < \epsilon^*(n,k|y_1^n),
	\label{rnepsilonstar}\\
R^*(n,\epsilon) 
&=&
	\frac{k}{n}\qquad  \text{iff} \qquad \epsilon^*(n,k+1) 
	\leq \epsilon < \epsilon^*(n,k).
	\label{rnepsilonstar2}
\ee
Therefore, we can concentrate on determining the fundamental
limits in terms of the rate; corresponding results for 
the minimal excess-rate probability then follow 
from~(\ref{rnepsilonstar}) and~(\ref{rnepsilonstar2}).

\medskip

\noindent
{\bf Prefix-free compressors. }
Let $R_p^*(n,\epsilon|y_1^n)$,
$R_p^*(n,\epsilon)$,
$\epsilon_p^*(n,k|y_1^n)$
and $\epsilon_p^*(n,k)$ be the corresponding
fundamental limits as those in 
Definitions~\ref{quenchedrndef}--\ref{epsilonstardef},
when the compressors are required to be prefix-free.
As it turns out, the prefix-free condition imposes 
a penalty of at most $1/n$ on the rate.
 
\begin{theorem} 
\label{prefixpenaltytheorem}
\begin{enumerate}
\item[$(i)$] For all $n,k \geq 1$:
\begin{equation}
\epsilon^*_p(n,k+1) = \begin{cases}
    \epsilon^*(n,k), &k<n\log{|\mathcal{X}|}\\
    0, &k\geq n\log{|\mathcal{X}|}.
  \end{cases}
\label{eq:nona}
\end{equation}
\item[$(ii)$]
For all $n \geq 1$ and any 
$0 \leq \epsilon < 1$:
\ben
R^*(n,\epsilon)\leq 
R^*_p(n,\epsilon) \leq R^*(n,\epsilon) + \frac{1}{n}.
\een
\end{enumerate}
\end{theorem}

Throughout the paper, `$\log$' denotes `$\log_2$', the logarithm taken
to base~2, and all familiar information theoretic quantities
are expressed in bits.

Note that,
for any fixed side information string $y_1^n$,
analogous results to those in Theorem~\ref{prefixpenaltytheorem}
hold for the reference-based
versions, $R_p^*(n,\epsilon|y_1^n)$ and $\epsilon^*_p(n,\epsilon|y_1^n)$,
as an immediate consequence of 
\cite[Theorem~1]{kontoyiannis-verdu:14} applied to the source with
distribution $\BBP(X_1^n=\cdot|Y_1^n=y_1^n)$.

\begin{proof}
For part $(i)$ note that, by the above remark, the 
reference-based analog of~(\ref{eq:nona})
in terms of 
$R_p^*(n,\epsilon|y_1^n)$ and $\epsilon^*_p(n,\epsilon|y_1^n)$
is an immediate consequence of 
\cite[Theorem~1]{kontoyiannis-verdu:14}.
Then,~(\ref{eq:nona}) follows directly by averaging
over all $y_1^n$.
The result of part~$(ii)$ follows directly from~$(i)$
upon noticing that
the analog of~(\ref{rnepsilonstar}) 
also holds for prefix-free codes:
$R_p^*(n,\epsilon) = \frac{k}{n}$
if and only if $\epsilon^*_p(n,k+1) \leq \epsilon < \epsilon^*_p(n,k).$
\end{proof}

\section{Direct and Converse Theorems for Arbitrary Sources}
\label{arbitrarysources}

In Section~\ref{s:general} we briefly describe generalisations
and extensions of the nonasymptotic coding theorems
in~\cite{kontoyiannis-verdu:14} to the case of compression
with side information. In Section~\ref{s:defns} we define
the core information-theoretic quantities that will be used
throughout the rest of the paper: The conditional information
density, and the conditional entropy and varentropy rates.

\subsection{General coding theorems}
\label{s:general}

Consider two arbitrary discrete random 
variables $(X,Y)$, with joint (PMF)
$P_{X,Y}$, taking values in $\clX$ and $\clY$, respectively.
For the sake of simplicity we may assume, without loss of generality,
that the source alphabet $\clX$ is the set of 
natural numbers $\clX=\mathbb{N}$, and that, for each $y\in\clY$, 
the values of $X$ are ordered with
nonincreasing conditional probabilities given $y$,
so that $\BBP(X=x|Y=y)$ is nonincreasing in $x$,
for each $y\in\clY$.

Let $f^*=f^*_1$ be the optimal compressor described in 
the last section, and write $P_X$ and
$P_{X|Y}$ for the PMF of $X$
and the conditional PMF
of $X$ given $Y$, respectively.  The ordering of the values 
of $X$ implies that, for all
$x\in\clX,y\in\clY$,
\be \ell(f^*(x|y))=\lfloor\log x\rfloor.
\label{eq:order}
\ee
 
The following 
is a general achievability result
that applies to both the reference-based
and the pair-based versions of the
compression problem:

\begin{theorem}
\label{t:achieve}
For all $x \in \mathcal{X}$, $y \in \mathcal{Y}$,
\begin{equation*}
\ell(f^*(x|y)) \leq -\log{P_{X|Y}(x|y)},
\end{equation*}
and for any $z\geq 0$,
$$\Pbig{\ell(f^*(X|Y)) \geq z} \leq \Pbig{-\log{P_{X|Y}(X|Y)} \geq z}.$$
\end{theorem}

The first part is an immediate consequence of 
\cite[Theorem~2]{kontoyiannis-verdu:14}, applied
separately for each $y\in\clY$ to the optimal compressor
$f^*(\cdot|y)$ for the source with distribution
$P_{X|Y}(\cdot|y)$. The second part follows directly
from the first.

The next theorem gives a general converse result
for the pair-based compression problem: 

\begin{theorem} \label{generalconvSI}
For any integer $k\geq 0$:
$$\mathbb{P}\left[\ell(f^{*}(X|Y)) \geq k\right]\geq
\sup_{\tau > 0}\big\{\mathbb{P}\left[-\log{P_{X|Y}(X|Y) \geq k + \tau} \right] 
- 2^{-\tau}\big\}.$$
\end{theorem}

\begin{proof}
Let $k\geq 0$ and $\tau > 0$ be arbitrary,
and define,
\ben
\mathcal{L} &=&
 \{(x,y) \in \mathcal{X}\times\mathcal{Y}: P_{X|Y}(x|y) \leq 2^{-k-\tau}\} \\
\mathcal{C} &=&
 \{(x,y) \in \mathcal{X}\times\mathcal{Y}: x \in \{1,2,\ldots,2^k-1\}\}.
\een
Then, 
\ben
\Pbig{-\log{P_{X|Y}(X|Y)} \geq k + \tau} 
&=& 
	P_{X,Y}(\mathcal{L}) \\
&=& 
	P_{X,Y}(\mathcal{L}\cap\mathcal{C}) + 
	P_{X,Y}(\mathcal{L}\cap\mathcal{C}^c) \\
&\leq& 
	P_{X,Y}(\mathcal{L}\cap\mathcal{C}) + P_{X,Y}(\mathcal{C}^c) \\
&\leq&
	\sum_{y \in \mathcal{Y}}{P_Y(y)\bigl((2^k-1)2^{-k-\tau}\bigr)} 
	+ \Pbig{\lfloor \log X\rfloor \geq k} \\
&\leq& 
	2^{-\tau}  + \Pbig{\ell(f^*(X|Y)) \geq k},
\een
where the last inequality follows from~(\ref{eq:order}).
\end{proof}

Our next result is one of the main tools in the proofs of the 
achievability results in the 
normal approximation bounds for $R^*(n,\epsilon|y_1^n)$
and $R^*(n,\epsilon)$
in Sections~\ref{memorylessnormalsection}
and~\ref{annealedsectionmemoryless}. It gives
tight upper and lower bounds on the performance of the
optimal compressor, that are useful in
both the reference-based and the pair-based setting:

\begin{theorem} 
\label{lemmalength}
For all $x,y$,
\be 
\ell(f^*(x|y)) 
&\geq& \log\left(\E\left[\left.\frac{1}{P_{X|Y}(X|y)}
\mathbb{I}_{ \{ P_{X|Y}(X|y) > P_{X|Y}(x|y)\}}\right|Y=y\right]\right) - 1,
	\label{eq:exact1}\\
 \ell(f^*(x|y)) 
&\leq& \log\left(\E\left[\left.\frac{1}{P_{X|Y}(X|y)}
\mathbb{I}_{ \{ P_{X|Y}(X|y) \geq P_{X|Y}(x|y)\}}\right|Y=y\right] \right),
\label{eq:exact2}
\ee
where $\IND_A$ denotes the indicator function of an event $A$,
with $\IND_A=1$ when $A$ occurs and $\IND_A=0$ otherwise.
\end{theorem}

\begin{proof}
Recall from~(\ref{eq:order}) that, for any
$k\in\NN,y\in\clY$, we have
$\ell(f^*(k|y))= \lfloor \log k \rfloor$.
In other words, for any $y$, the optimal description length
of the $k$th most likely value of $X$ according to
$P_{X|Y}(\cdot|y)$, is $\lfloor \log k \rfloor$ bits.
Although there may be more than one optimal ordering 
of the values of $X$ when there are ties, 
it is always the case that (given $y$)
the position of $x$ is 
between the number of values that have probability 
strictly larger than $P_{X|Y}(x|y)$ and the number of 
outcomes that have probability $\geq P_{X|Y}(x|y)$.
Formally,
$$\left\lfloor\log{\left( \sum_{x' \in \mathcal{X}}
{\mathbb{I}_{ \{ P_{X|Y}(x'|y) > P_{X|Y}(x|y)\}}} \right)}\right\rfloor 
\leq \ell(f^*(x|y) ) \leq 
\left\lfloor\log{\left( \sum_{x' \in \mathcal{X}}
{\mathbb{I}_{ \{ P_{X|Y}(x'|y) \geq P_{X|Y}(x|y)\}}} \right)} \right\rfloor.
$$
Multiplying and dividing each summand above
by $P_{X|Y}(x|y)$, the two sums are equal to the
expectations in~(\ref{eq:exact1})
and~(\ref{eq:exact2}), respectively.
The result follows from the trivial bounds on the floor function,
$a\geq\lfloor a\rfloor\geq a-1$.
\end{proof}

\subsection{Conditional entropy and varentropy rate}
\label{s:defns}

Suppose $(\Xp,\Yp)=\{(X_n,Y_n)\;;\;n\geq 1\}$ is an arbitrary
source-side information pair,
with values in the finite alphabets
$\clX,\clY$. 

\begin{definition}[Conditional information density] $\;$
For any source-side information pair $(\Xp,\Yp)$, the
{\em conditional information density} of blocklength~$n$
is the random variable,
$$-\log P(X_1^n|Y_1^n)=-\log P_{X_1^n|Y_1^n}(X_1^n|Y_1^n).$$
\end{definition}

\noindent
When it causes no confusion, we drop the 
subscripts for PMFs and conditional PMFs,
e.g., simply writing $P(x_1^n|y_1^n)$ for 
$P_{X_1^n|Y_1^n}(x_1^n|y_1^n)$ 
above.
Throughout, $H(Z)$ and $H(Z|W)$ denote the discrete entropy
of $Z$ and the conditional entropy of $Z$ given $W$,
in bits.

\begin{definition}[Conditional entropy rate]
For any source-side information pair $(\Xp,\Yp)$, the
{\em conditional entropy rate} $H(\Xp|\Yp)$ is defined
as:
$$
H(\Xp|\Yp)=\limsup_{n\to\infty}\frac{1}{n}H(X_1^n|Y_1^n).$$
\end{definition}

\noindent
If $(\Xp,\Yp)$ are jointly stationary, then the above 
$\limsup$ is in fact a limit,
and it is equal to $H(\Xp,\Yp)-H(\Yp)$,
where $H(\Xp,\Yp)$ and $H(\Yp)$
are the entropy rates of $(\Xp,\Yp)$ and
of $\Yp$, respectively \cite{cover:book2}.

\begin{definition}[Conditional varentropy rate]
\label{asymptoticvariance}
\hspace{-0.11in}
For any source-side information pair $(\Xp,\!\Yp)$, 
the {\em conditional varentropy rate} is:
\begin{equation}
\label{variancedef}
\sigma^2(\Xp|\Yp) 
= \limsup_{n \rightarrow \infty}{\frac{1}{n}
\VAR{\left(-\log{P(X_1^n|Y_1^n)}\right)}}.
\end{equation}
\end{definition}

\medskip

\noindent
As with the conditional entropy rate,
under additional assumptions the $\limsup$ in~(\ref{variancedef})
is in fact a limit. Lemma~\ref{lem:varentropy} is proved in 
Appendix~\ref{app:Mvar}.

\begin{lemma}
\label{lem:varentropy}
If the source-side information pair $(\Xp,\Yp)$
and the side information process $\Yp$,
are irreducible and aperiodic, $d$th order Markov chains,
then,
$$H(\Xp|\Yp)=\lim_{n\to\infty}\frac{1}{n}H(X_1^n|Y_1^n),$$
and
\begin{equation*}
\sigma^2(\Xp|\Yp) = \lim_{n \rightarrow \infty}{\frac{1}{n}
\VAR{\left(-\log{P(X_1^n|Y_1^n)}\right)}}.
\end{equation*}
In particular, the limits 
$H(\Xp|\Yp)$ and $\sigma^2(\Xp|\Yp)$
exist and they are independent
of the initial distribution of the chain $(\Xp,\Yp)$.
\end{lemma}

\section{Normal Approximation for Reference-Based Compression}
\label{memorylessnormalsection}

In this section we give explicit, finite-$n$ bounds
on the reference-based optimal rate $R^*(n,\epsilon|y_1^n)$.
Suppose the source and side information, $(\Xp,\Yp)$,
consist of independent
and identically distributed (i.i.d.) pairs $\{(X_n,Y_n)\}$,
or, more generally, that $(\Xp,\Yp)$ is a {\em conditionally-i.i.d.\
source-side information pair}, i.e., that the distribution 
of $\Yp$ is arbitrary, and for each $n$, given $Y_1^n=y_1^n$,
the random variables $X_1^n$ are conditionally i.i.d.,
$$\BBP(X_1^n=x_1^n|Y_1^n=y_1^n)=\prod_{i=1}^nP_{X|Y}(x_i|y_i),
\qquad x_1^n\in\clX^n,y_1^n\in\clY^n,$$
for a given family of conditional
PMFs $P_{X|Y}(\cdot|\cdot)$.

We will use the following notation. 
For any $y\in\clY$, we write,
$$H(X|y)=-\sum_{x\in\clX} P_{X|Y}(x|y)\log P_{X|Y}(x|y),$$ 
for the entropy of the conditional distribution of $X$
given $Y=y$, and,
\begin{equation}
V(y) = \VAR[-\log{P_{X|Y}(X|y)}|Y=y].
\label{eq:Vdef}
\end{equation}
For a side information string $y_1^n\in\clY^n$,
we denote,
 \be 
	H_n(X|y_1^n)&=&\frac{1}{n}\sum_{j=1}^nH(X|y_j)
	\label{eq:Hndef}\\
	\sigma^2_n(y_1^n) &=& \frac{1}{n}\sum_{j=1}^n{V(y_j)}.
	\label{eq:sndef}
\ee

The upper and lower bounds developed in 
Theorems~\ref{quenchedupper}
and~\ref{quenchedconverse}
below say that, for any conditionally-i.i.d.\
source-side information pair
$(\Xp,\Yp)$ and any side information 
string $y_1^n$, the reference-based
optimal compression rate,
\be
R^*(n,\epsilon|y_1^n)= H_n(X|y_1^n)
+\frac{\sigma_n(y_1^n)}{\sqrt{n}}
Q^{-1}(\epsilon)-\frac{\log n}{2n} +O\Big(\frac{1}{n}\Big),
\label{eq:rateapprox}
\ee
with explicit bounds on the $O(1/n)$ term,
where $Q$ denotes the standard Gaussian 
tail function 
$Q(x)=\frac{1}{\sqrt{2\pi}}\int_x^\infty e^{-z^2/2}dz$.

As described in the Introduction, 
$R^*(n,\epsilon|y_1^n)$ is the best achievable
rate with excess-rate probability no more than
$\epsilon$, with respect to a {\em fixed} side 
information string $y_1^n$. 

\subsection{Preliminaries}
\label{s:prelim}

Suppose for now that $(\Xp,\Yp)=\{(X_n,Y_n)\}$ is
an i.i.d.\ source-side information pair, with
all $(X_n,Y_n)$ distributed as $(X,Y)$, 
with joint PMF $P_{X,Y}$ on
$\clX\times\clY$. 
In this case, the conditional entropy 
rate $H(\Xp|\Yp)$ is simply $H(X|Y)$ and
the conditional varentropy rate~(\ref{variancedef}) 
reduces to the \textit{conditional varentropy} of 
$X$ given $Y$,
\begin{equation*}
\sigma^2(X|Y) = \VAR{\bigl[-\log{P_{X|Y}(X|Y)}\bigr]},
\end{equation*}
where $P_{X|Y}$ denotes the conditional
PMF of $X_n$ given $Y_n$.
As in earlier sections, we will drop the subscripts of
PMFs when they can be understood unambiguously from the context. 

First we state some simple properties for 
the conditional varentropy. We write $\hat{H}_X(Y)$ 
for the random variable,
\be
\hat{H}_X(Y)=-\sum_{x\in\clX} 
P_{X|Y}(x|Y)\log P_{X|Y}(x|Y).
\label{eq:Hhatdef}
\ee

\begin{proposition}
\label{varentropyprop} 
Suppose $(\Xp,\Yp)$ is an i.i.d.\ source-side 
information pair, with each $(X_n,Y_n)\sim(X,Y)$. Then:
\begin{enumerate}
\item[$(i)$]
\label{varentropypropfirst}
The conditional varentropy can also be expressed:
\begin{equation*} 
\sigma^2(X|Y) = \BBE[V(Y)]
 + \VAR[\hat{H}_X(Y)].
\end{equation*}
\item[$(ii)$]
$\BBE[V(Y)]=0$ if and only if, for each $y \in \mathcal{Y}$, 
$P_{X|Y}(x|y)$ is uniform on a (possibly singleton)
subset of $\mathcal{X}$.
\item[$(iii)$]
$\sigma^2(X|Y) = 0$
if and only if 
there exists $k \in \{1,2,\ldots,|\mathcal{X}|\}$, 
such that, for each $y \in \mathcal{Y}$, 
$P_{X|Y}(x|y)$ is uniform on a 
subset $\clX_y\subset\mathcal{X}$
of size $|\clX_y|=k$.
\end{enumerate}
\end{proposition}

\begin{proof}
For~$(i)$ we have,
\begin{align*}
\sigma^2(X|Y) &= \VAR[-\log P(X|Y)] \\
&= \E[(\log P(X|Y))^2] - H(X|Y)^2 \\
&= \E[(\log{P(X|Y))^2} 
	- \E[\hat{H}_X(Y)^2] 
	+ \E[\hat{H}_X(Y)^2] 
	- H(X|Y)^2 \\
&= \E\big\{
	\E[(\log P(X|Y))^2|Y] - \hat{H}_X(Y)^2 
	\big\}
 	+ \VAR[\hat{H}_X(Y)]\\
&= \E[V(Y)]
 	+ \VAR[\hat{H}_X(Y)].
\end{align*}
Parts~$(ii)$ and~$(iii)$ are
straightforward from the definitions.
\end{proof}

\subsection{Direct and converse bounds}

Before stating our main results we note that,
if $\sigma_n^2(y_1^n)$ were equal to zero 
for some side information sequence $y_1^n$, 
then each source symbol would be known 
(both to the encoder and decoder) to be 
uniformly distributed on some subset of $\clX$,
so the compression problem would be rather 
trivial. To avoid these degenerate cases,
we assume that $\sigma_n^2(y_1^n)>0$ in
Theorems~\ref{quenchedconverse}
and~\ref{quenchedupper}.

\begin{theorem}[Reference-based converse]
\label{quenchedconverse}
Suppose $(\Xp,\Yp)$ is a conditionally-i.i.d.\ source-side
information pair.
For any $0 < \epsilon < \frac{1}{2},$
the reference-based optimal compression rate satisfies,
\begin{equation}
\label{lowerboundfixed}
R^*(n,\epsilon|y_{1}^{n}) 
\geq H_n(X|y_1^n) 
+ \frac{\sigma_{n}(y_1^n)}{\sqrt{n}}Q^{-1}(\epsilon) 
-\frac{\log{n}}{2n} - \frac{1}{n}\eta(y_{1}^{n}),
\end{equation}
for all,
\begin{equation}
\label{lbthres}
n > 
\frac{(1+6m_{3}\sigma_n^{-3}(y_1^n))^{2}}
{4\bigl(Q^{-1}(\epsilon)\phi(Q^{-1}(\epsilon))\bigr)^{2}},
\end{equation} 
and any side information string $y_1^n\in\clY^n$
such that $\sigma_n^2(y_1^n)>0$,
where $\phi$ is the standard normal density,
$H_n(X|y_1^n)$ and $\sigma_n^2(y_1^n)$ are
given in~{\em (\ref{eq:Hndef})} and~{\em (\ref{eq:sndef})},
\be
m_{3} = \max_{y \in \mathcal{Y}}
\E\Big[\big|-\log{P(X|y)}-H(X|y)\big|^{3}\Big|Y=y\Big],
\label{eq:m3}
\ee
and,
\begin{equation*}
\eta(y_{1}^{n}) 
= \frac{\sigma_n^3(y_1^n) + 6m_{3}}
{\phi(Q^{-1}(\epsilon))\sigma_n^2(y_1^n)}.
\end{equation*}
\end{theorem}

Note that, by the definitions in Section~\ref{descriptionofanoptimal},
Theorem~\ref{quenchedconverse}
obviously also holds for prefix-free codes,
with $R_p^*(n,\epsilon|y_1^n)$ in place of $R^*(n,\epsilon|y_1^n)$.

\begin{proof}
Since, conditional on $y_{1}^{n}$, 
the random variables $X_1^n$ are independent,
we have,
\begin{align}
& \PP\left.\left[ -\log{P(X_{1}^{n}|y_{1}^{n})} 
\geq 
\sum_{i=1}^{n}{H(X|y_{i})} + \sqrt{n}\sigma_n(y_1^n)Q^{-1}(\epsilon) 
- \eta(y_{1}^{n})\right|Y_1^n = y_1^n\right] \nonumber\\
& = \PP\biggl[ 
\left.\frac{\sum_{i=1}^{n}{(-\log{P(X_{i}|y_{i})} - H(X|y_{i}))}}
{\sigma_n(y_1^n)\sqrt{n}} \geq Q^{-1}(\epsilon) 
- \frac{\eta(y_{1}^{n})}{\sigma_n(y_1^n)\sqrt{n}}
\right|Y_1^n = y_1^n\biggr] \nonumber\\ 
&\geq \label{BEjustify} 
Q\left(Q^{-1}(\epsilon)-\frac{\eta(y_{1}^{n})}{\sigma_n(y_1^n)\sqrt{n}}\right) 
- 6\frac{m_{3}}{\sigma_n^3(y_1^n)\sqrt{n}} \\ 
&\geq \label{MVTjustify} \epsilon + \frac{\eta(y_{1}^{n})}
{\sigma_n(y_1^n)\sqrt{n}}\phi(Q^{-1}(\epsilon)) 
- 6\frac{m_{3}}{\sigma_n^3(y_1^n)\sqrt{n}} \\ 
&= \epsilon + \frac{1}{\sqrt{n}},
\label{eq:eta}
\end{align}
where~(\ref{BEjustify}) follows from 
the Berry-Ess\'{e}en bound~\cite{fellerII:book},
(\ref{MVTjustify}) follows from the fact that
a simple first-order Taylor expansion gives,
\begin{equation} \label{QineqConv}
Q(\alpha-\Delta) - Q(\alpha ) \geq \Delta\phi(\alpha), 
\qquad\mbox{for}\;
0< \alpha < \frac{\Delta}{2},
\end{equation}
and~(\ref{eq:eta}) follows from the definition
of $\eta(y_1^n)$.
Putting $\alpha = Q^{-1}(\epsilon)$ and 
$\Delta = \Delta(y_{1}^{n}) = \frac{\eta(y_{1}^{n})}{\sqrt{s_{n}}}$, 
(\ref{lbthres})~is sufficient 
for (\ref{QineqConv}) to hold.

Since we condition on the fixed side information sequence $y_{1}^{n}$, 
\cite[Theorem~4]{kontoyiannis-verdu:14} applies,
with $\tau = \frac{1}{2}\log{n}$, where we replace $X$ 
by $X_{1}^{n}$ with PMF $P_{X_1^n|Y_1^n}(\cdot|y_{1}^{n})$. 
Thus, putting,
$K_{n}=\sum_{i=1}^{n}{H(X|y_{i})} + 
\sigma_n(y_1^n)\sqrt{n}Q^{-1}(\epsilon) - \eta(y_{1}^{n}),$
yields,
$$
\PP\left.\left[ \ell(f_n^{*}(X_{1}^{n}|y_{1}^{n})) 
\geq  K_{n} -\frac{\log{n}}{2}\right|Y_1^n = y_1^n  \right] 
\geq \PP\big[ -\log{P(X_{1}^{n}|y_{1}^{n})} \geq K_{n}
\big|Y_1^n = y_1^n\big] - \frac{1}{\sqrt{n}}
\geq \epsilon,
$$
and the claimed bound follows.
\end{proof}

\medskip

Although the expressions~(\ref{lowerboundfixed})
and~(\ref{lbthres}) in
Theorem~\ref{quenchedconverse}
are quite involved, we note that their purpose
is not to be taken as exact values used in practice.
Instead, their utility is to first demonstrate that
finite-$n$ performance guarantees that
are accurate up to $O(1/n)$ terms in the rate
are indeed possible to provide,
and to illustrate the nature of the
dependence of the rate and the minimal blocklength
on the problem parameters. Similar comments
apply to the constants in Theorems~\ref{quenchedupper},
\ref{IIDAnnealedAchievability},
and~\ref{IIDAnnealedConverse} below. Nevertheless,
the actual numerical values obtained are in many
cases quite realistic, as illustrated by the example
in the remark following 
Theorem~\ref{quenchedupper}.

Next we derive an upper bound to $R^*(n,\epsilon|y_1^n)$
that matches the lower bound
in Theorem~\ref{quenchedconverse} up to and including
the third-order term. Note that, in view 
of Theorem~\ref{prefixpenaltytheorem},
the result of Theorem~\ref{quenchedupper} also holds
for prefix-free codes, with $R_p^*(n,\epsilon|y_1^n)$ and
$\zeta_n(y_1^n)+1$ in place of 
$R^*(n,\epsilon|y_1^n)$ and
$\zeta_n(y_1^n)$, respectively.

\begin{theorem}[Reference-based achievability]
\label{quenchedupper}
Let $(\Xp,\Yp)$ be a conditionally-i.i.d.\ source-side
information pair.
For any $0 < \epsilon \leq \frac{1}{2},$
the reference-based optimal compression rate satisfies,
$$
R^*(n,\epsilon|y_1^n) \leq H_n(X|y_1^n) + \frac{\sigma_n(y_1^n)}{\sqrt{n}}
Q^{-1}(\epsilon) - \frac{\log{n}}{2n} +\frac{1}{n}\zeta_n(y_1^n),
$$
for all,
\begin{equation} 
\label{nthresupperquenched}
n > \frac{36m_3^2}{\epsilon^2\sigma_n^6(y_1^n)},
\end{equation}
and any side information string 
$y_1^{n} \in \mathcal{Y}^{n}$ 
such that $\sigma_n^2(y_1^n)>0$,
where 
$H_n(X|y_1^n)$ and $\sigma_n^2(y_1^n)$ are
given in~{\em (\ref{eq:Hndef})} and~{\em (\ref{eq:sndef})},
$m_3$ is given in~{\em (\ref{eq:m3})}, and,
\be
\zeta_n(y_1^n)
=
\frac{6m_3}
{\sigma_n^3(y_1^n)
\phi\Big(\Phi^{-1}\Big(\Phi(Q^{-1}(\epsilon)) + \frac{6m_3}
{\sqrt{n}\sigma_n^3(y_1^n)}\Big)\Big)} 
+ 
\log{\Bigl( \frac{\log{e}}{\sqrt{2\pi\sigma_n^2(y_1^n) }} 
+ \frac{12m_3}{\sigma_n^3(y_1^n)}\Bigr)}.
\label{eq:zeta}
\ee
\end{theorem}

\medskip

\noindent
{\bf Remark. } 
Before giving the proof we note that, although as mentioned
earlier the values of the constants appearing in the theorem
are more indicative of general trends than actual realistic
values, in most cases they do give reasonable estimates.
For example, consider the lower bound in~(\ref{nthresupperquenched})
on the blocklength required for Theorem~\ref{quenchedupper} to hold.

In order to provide a fair illustration, as an example we 
took a {\em random} joint distribution 
for $(X,Y)$, assuming they both take values in a four-letter alphabet,
and obtained:
$$
P_{XY}=
\left( \begin{array}{cccc}
    0.0195 &  0.0641 &  0.0078 &  0.0228\\
    0.0624 &  0.3250 &  0.0080 &  0.1196\\
    0.0481 &  0.0836 &  0.0043 &  0.0110\\
    0.0543 &  0.0442 &  0.0099 &  0.1154
\end{array} \right).
$$
Suppose that the side information 
string $y_1^n$ is not ``too atypical,''
in that its empirical frequencies
$(0.2,0.55,0.05,0.2)$ are not too far
from 
$P_Y=(0.1843,  0.5169, 0.0301, 0.2687)$,
the true marginal of $Y$.
Then 
the lower bound 
in~(\ref{nthresupperquenched})
is in fact a little better than,
$$n\geq \frac{1950}{\epsilon^2}.$$
This is 
intuitively satisfying as it shows that (1.)
the smaller the error probability $\epsilon$
is required to be, the larger the necessary blocklength
for rate guarantees; and (2.) the dependence
of the blocklength on $\epsilon$ is of $O(1/\epsilon^2)$.

Although the above example was for a randomly selected joint
distribution, there are many cases where the bound
in~(\ref{nthresupperquenched}) is in fact much more 
practical.
For example, taking $X,Y$
with values in $\{1,2,\ldots,100\}$, choosing
for the marginal of $Y$, $P_Y(k)\propto 2^{-k}$,
and for the conditional $P_{X|Y}(n|k)\propto n^{-k}$,
then for any typical side information string $y_1^n$,
we obtain:
$$n\geq \frac{6}{\epsilon^2}.$$

\begin{proof}
Given $\epsilon$ and $y_1^n$, let $\beta_n = \beta_n(y_1^n)$ 
be the unique constant such that,
\begin{align} 
\label{betandef1}
\Pbig{-\log{P(X_1^n|y_1^n)} \leq \log{\beta_n}|Y_1^n = y_1^n} &\geq 1-\epsilon, \\
\Pbig{-\log{P(X_1^n|y_1^n)} < \log{\beta_n}|Y_1^n = y_1^n} &< 1-\epsilon,
\nonumber
\end{align}
and write $\lambda_n$ for its normalised
version,
\begin{equation*}
\lambda_n=\frac{\log{\beta_n} - \sum_{i=1}^n{H(X|y_i)}}
{\sqrt{n}\sigma_n(y_1^n)}.
\end{equation*}
Using 
the Berry-Ess\'{e}en bound 
\cite{fellerII:book}
yields,
\begin{equation*}
1-\epsilon 
\leq 
\BBP\left.\left[\frac{-\log{P(X_1^n|y_1^n)}-\sum_{i=1}^n{H(X|y_i)}}
{\sigma_n(y_1^n)\sqrt{n}} \leq \lambda_n\right|Y_1^n = y_1^n\right]
\leq \Phi(\lambda_n) + \frac{6m_3}{\sigma_n^3(y_1^n)\sqrt{n}},
\end{equation*}
and,
\begin{equation}\label{sandwichlambdan2}
1-\epsilon > 
\BBP\left.\left[\frac{-\log{P(X_1^n|y_1^n)} - \sum_{i=1}^n{H(X|y_i)}}
{\sigma_n(y_1^n)\sqrt{n}} < \lambda_n\right|Y_1^n = y_1^n\right]
\geq \Phi(\lambda_n) - \frac{6m_3}{\sigma_n^3(y_1^n)\sqrt{n}}.
\end{equation}
Define,
\begin{equation*}
\lambda = \Phi^{-1}(1-\epsilon) = Q^{-1}(\epsilon).
\end{equation*}
For $n$ satisfying~(\ref{nthresupperquenched}), we have,
$$\Phi(\lambda) + \frac{6m_3}{\sigma_n^3(y_1^n)\sqrt{n}} < 1,$$
so, using~(\ref{sandwichlambdan2}) and a 
first-order Taylor expansion,
we obtain,
\begin{align}
\lambda_n 
&\leq \Phi^{-1}\Big(\Phi(\lambda) + 
\frac{6m_3}{\sigma_n^3(y_1^n)\sqrt{n}} \Big) 
\nonumber\\
&=\lambda 
+ \frac{6m_3}{\sigma_n^3(y_1^n)\sqrt{n}}(\Phi^{-1})'(\xi_n)
\nonumber\\
&=\label{lambdanbound} \lambda 
+ \frac{6m_3}{\sigma_n^3(y_1^n)\sqrt{n}\phi(\Phi^{-1}(\xi_n))},
\end{align}
for some $\xi_n = \xi_n(y_1^n)$ between
$\Phi(\lambda)$ and 
$\Phi(\lambda) + 6m_3/\sigma_n^3(y_1^n)\sqrt{n}$. 

\newpage

Since $\epsilon \leq \frac{1}{2},$ we have 
$\lambda \geq 0$ and $\Phi(\lambda) \geq \frac{1}{2},$ 
so that $\xi_n \geq \frac{1}{2}.$ Also, since $\Phi^{-1}(t)$ is strictly 
increasing for all $t$ and $\phi$ is strictly decreasing for $t \geq 0$, 
from~(\ref{lambdanbound}) we get, 
\begin{equation}
\label{lambdanfinalbound}
\lambda_n \leq \lambda + \frac{6m_3}
{\sigma_n^3(y_1^n)\sqrt{n}}\times
\frac{1}{\phi\Big(\Phi^{-1}\Big(\Phi(\lambda) 
+ \frac{6m_3}{\sigma_n^3(y_1^n)\sqrt{n}}\Big)\Big)}.
\end{equation} 

On the other hand, from the discussion 
in the proof of Theorem~\ref{lemmalength}, 
together with (\ref{betandef1}), we conclude that,
\begin{equation*}
\BBP
\left.\left[
\ell(f^*_n(X_1^n|y_1^n))>\log\left(\sum_{x_1^n \in \mathcal{X}^n}
{\mathbb{I}_{\{P(x_1^n|y_1^n) \geq \frac{1}{\beta_n}\}}}\right)
\right|Y_1^n=y_1^n\right] \leq \epsilon,
\end{equation*}
hence,
\begin{align}
R^*(n,\epsilon|y_1^n) 
&\leq 
\frac{1}{n}\log\left(\sum_{x_1^n \in \mathcal{X}^n}
{\mathbb{I}_{\{P(x_1^n|y_1^n) \geq \frac{1}{\beta_n}\}}}\right)
\nonumber\\
&=\frac{1}{n}
\log \left( \BBE\left[\left.
2^{-\log{P(X_1^n|y_1^n)}}\mathbb{I}_{\{-\log{P(X_1^n|y_1^n)}
\leq \log{\beta_n}\}}\right|Y_1^n = y_1^n\right] \right) 
\nonumber\\ 
\label{Rnyfinal}
&= \frac{1}{n}\sum_{i=1}^n{H(X|y_i)} 
+ \lambda_n\frac{\sigma_n(y_1^n)}{\sqrt{n}} + \frac{1}{n}\log{\alpha_n},
\end{align}
where,
\begin{align*}
\alpha_n 
&= \BBE
\left.\left[
2^{-\log{\beta_n}-\log{P(X_1^n|y_1^n)}}
\mathbb{I}_{\{\log{\beta_n}+\log{P(X_1^n|y_1^n)} \geq 0\}}
\right|Y_1^n = y_1^n\right]
\\ 
&= \BBE\left[\left.
2^{-\sigma_n(y_1^n)\sqrt{n}(\lambda_n - Z_n)}
\mathbb{I}_{\{\sigma_n(y_1^n)\sqrt{n}(\lambda_n - Z_n) \geq 0\}}
\right|Y_1^n = y_1^n\right],
\end{align*}
and,
\ben
Z_n &=&
	\frac{1}{\sigma_n(y_1^n)\sqrt{n}}
	\left[-\log{P(X_1^n|y_1^n)} - \sum_{i=1}^n{H(X|y_i)}\right]\\
&=&
	\frac{1}{\sigma_n(y_1^n)\sqrt{n}}
	\sum_{i=1}^n{\big[-\log{P(X_i|y_i) - H(X|y_i)}\big]}. 
\een
Note that $Z_n$ has zero mean and unit variance. 
Let,
\begin{equation*}
\bar{\alpha}_n = 
\BBE\left(
2^{-\sigma_n(y_1^n)\sqrt{n}(\lambda_n - Z)}
\mathbb{I}_{\{\sigma_n(y_1^n)\sqrt{n}(\lambda_n - Z) \geq 0\}}
\right),
\end{equation*} 
where $Z$ is a standard normal random variable.
Then,
$$\bar{\alpha}_n = \int_{0}^{\infty}
\frac{1}{\sqrt{2\pi n\sigma^2_n(y_1^n)}}\,
2^{-x}\,
\exp\Big\{-\frac{(x-\lambda_n\sigma_n(y_1^n)\sqrt{n})^2}
{2n\sigma_n^2(y_1^n)}\Big\}\,dx 
\leq \frac{\log{e}}{\sqrt{2\pi n\sigma_n^2(y_1^n)}}.
$$
Denoting by $F_n(t)$ the distribution function
of $Z_n$, and integrating by parts,
\begin{align}
\alpha_n 
&= 
	\int_{-\infty}^{\lambda_n}
	2^{-\sigma_n(y_1^n)\sqrt{n}(\lambda_n-t)}dF_n(t)
	\nonumber\\
&= 
	F_n(\lambda_n) 
	- (\log_e{2})\int_{-\infty}^{\lambda_n}{F_n(t)\sigma_n(y_1^n)\sqrt{n}
	2^{-\sigma_n(y_1^n)\sqrt{n}(\lambda_n-t)}dt} 
	\nonumber\\
&= 
	\bar{\alpha}_n + F_n(\lambda_n) - \Phi(\lambda_n) 
	-(\log_e{2})\sigma_n(y_1^n)\sqrt{n}
	\int_{-\infty}^{\lambda_n}{(F_n(t)-\Phi(t))
	2^{-\sigma_n(y_1^n)\sqrt{n}(\lambda_n-t)}dt}  
	\nonumber\\
&\leq 
	\bar{\alpha}_n + \frac{6m_3}{\sqrt{n}\sigma_n^3(y_1^n)} 
	+ \frac{6m_3}{\sqrt{n}\sigma_n^2(y_1^n)}
	(\log_e{2})
	\int_{-\infty}^{\lambda_n}
	{2^{-\sigma_n(y_1^n)\sqrt{n}(\lambda_n-t)}dt} 
	\nonumber\\
&\leq 
	\bar{\alpha}_n + \frac{12m_3}{\sqrt{n}\sigma_n^3(y_1^n)} 
	\nonumber
	\\ 
\label{anfinalbound}
&\leq 
	\frac{1}{\sqrt{n}}
	\Bigl( \frac{\log{e}}{\sqrt{2\pi \sigma_n^2(y_1^n)}} 
	+ \frac{12m_3}{\sigma_n^3(y_1^n)}\Bigr),
\end{align}
where we used 
the Berry-Ess\'{e}en bound~\cite{fellerII:book}
twice.

The claimed bound follows from~(\ref{lambdanfinalbound}),
(\ref{Rnyfinal}), and~(\ref{anfinalbound}).
\end{proof}

\section{Normal Approximation for Pair-Based Compression}
\label{annealedsectionmemoryless}

Here we give upper and lower bounds to the pair-based optimal
compression rate $R^*(n,\epsilon)$, analogous to those
presented in Theorems~\ref{quenchedconverse}
and~\ref{quenchedupper} for the 
reference-based optimal rate: For an i.i.d.\ source-side
information pair
$(\Xp,\Yp)$, the result in Theorems~\ref{IIDAnnealedAchievability}
and~\ref{IIDAnnealedConverse} state that,
\be
R^*(n,\epsilon)=
H(X|Y) + \frac{\sigma(X|Y)}{\sqrt{n}}Q^{-1}(\epsilon) 
-\frac{\log{n}}{2n} + O\Big(\frac{1}{n}\Big),
\label{eq:pbrough}
\ee
with explicit upper and lower bounds for the
$O(1/n)$ term.

Recall the discussion in the Introduction
comparing~(\ref{eq:pbrough}) with the corresponding
expansion~(\ref{eq:rateapprox}) in the reference-based case.
In particular, we note that, for large $n$, we typically
have 
$H_n(X|y_1^n)\approx H(X|Y)$, but
$\sigma_n^2(y_1^n)<\sigma^2(X|Y)$.

Unlike our achievability result for the reference-based rate,
the corresponding expansion for the pair-based rate 
requires a very different approach from that in the 
case without side information. The main step in the proof
of the corresponding result, given 
in~\cite[Eq.~(167)]{kontoyiannis-verdu:14}) 
does not generalise to the side-information setting.
Instead, we use our new Theorem~\ref{lemmalength} as 
the main approximation tool. 

\begin{theorem}[Pair-based achievability]
\label{IIDAnnealedAchievability}
Let $(\Xp,\Yp)$ be an i.i.d.\ source-side
information pair, with conditional varentropy rate
$\sigma^2=\sigma^2(X|Y)>0$.
For any $0 < \epsilon \leq \frac{1}{2},$
the pair-based optimal compression rate satisfies,
\begin{equation}
R^*(n,\epsilon) \leq H(X|Y) + \frac{\sigma(X|Y)}{\sqrt{n}}Q^{-1}(\epsilon) 
-\frac{\log n}{2n} + \frac{C}{n},
\label{eq:iidachieve}
\end{equation}
for all,
\begin{equation} 
\label{nthrescondrateub}
n > 
\frac{4\sigma^2}{B^2\phi(Q^{-1}(\epsilon))^2}
\times\left[
\frac{B^2}
{2\sqrt{2\pi e}\sigma^2} 
+  \frac{\psi^2}{(1-\frac{1}{2\pi})^2\bar{v}^2}\right]^2,
 \end{equation}
where $\bar{v}=\BBE[V(Y)]$ and $\psi^2=\VAR(V(Y))$,
with $V$ defined in~{\em (\ref{eq:Vdef})},
\begin{equation*} 
C = \log
{\Bigl(
\frac{2}{\bar{v}^{1/2}} 
+ \frac{24m_3(2\pi)^{3/2}}{\bar{v}^{3/2}}  
\Bigr)} 
+ B,
\end{equation*}
$m_3$ is given in~{\em (\ref{eq:m3})}, 
and,
$$B=
\frac{\E\big[|-\log{P(X|Y)}-H(X|Y)\big|^3]}{\sigma^2\phi(Q^{-1}(\epsilon))}.$$
\end{theorem}

As in Section~\ref{memorylessnormalsection},
we note that,
in view 
of Theorem~\ref{prefixpenaltytheorem},
the result of Theorem~\ref{IIDAnnealedAchievability}
remains true for prefix-free codes, 
with $R_p^*(n,\epsilon)$ and
$C+1$ in place
$R^*(n,\epsilon)$ and $C$, respectively.

\begin{proof}
We will use the achievability part of
Theorem~\ref{lemmalength}.
For each $i$, take $X_i$ and $\bar{X}_i$ to be conditionally independent 
versions of $X$ given $Y=Y_i$, i.e., assume that
$P(X_i,\bar{X}_i|Y_i)=P(X_i|Y_i)P(\bar{X}_i|Y_i)$,
where the pair $(\bar{X}_i,Y_i)$ has the same distribution
as $(X_i,Y_i)$.
Define,
\begin{align*}
S_n
&= \frac{1}{\sigma_n(Y_1^n)\sqrt{n}}
\sum_{i=1}^{n}{[-\log{P(X_i|Y_i) - \hat{H}_X(Y_i)]}},\\
\bar{S}_n 
&= \frac{1}{\sigma_n(Y_1^n)\sqrt{n}}
\sum_{i=1}^{n}{[-\log{P(\bar{X}_i|Y_i) - \hat{H}_X(Y_i)]}},\\
T_n
& = \frac{1}{\sqrt{n}}
\sum_{i=1}^{n}{[\hat{H}_X(Y_i) - H(X|Y)]},
\end{align*}
where $\sigma^2_n(Y_1^n)$, $\hat{H}_X(Y)$
are defined in~(\ref{eq:sndef}) and~(\ref{eq:Hhatdef}), respectively.
For any $K>0$ the upper bound in 
Theorem~\ref{lemmalength} gives,
\begin{align*}
&
	\BBP\left[\ell(f_n^*(X_1^n|Y_1^n)) > K\right]\\
&
	\leq \BBP\left[
	\E\left(\left. \frac{1}{P(\bar{X_1^n}|Y_1^n)}
	\mathbb{I}_{ \{ P(\bar{X}_1^n|Y_1^n) 
	\geq P(X_1^n|Y_1^n)\}}\right|X_1^n,Y_1^n\right) > 2^{K} \right]\\
&  
	=    \BBP \left[
	\E\left(\left. 2^{\sqrt{n}\sigma_n(Y_1^n)\bar{S}_n}
	\mathbb{I}_{\{\bar{S}_n \leq S_n\}}
	\right|X_1^n,Y_1^n\right) 
	> 2^{K - \sum_{i=1}^n \hat{H}_X(Y_i)}\right],
\end{align*}
and taking
$K=K_n=nH(X|Y) + \sigma(X|Y)\sqrt{n}Q^{-1}(\epsilon) -\log{\sqrt{n}} + C$,
\begin{align}
&
	\BBP\left[\ell(f_n^*(X_1^n|Y_1^n)) > K_n\right]
	\nonumber\\
&
	\leq \BBP\left[
	\E\left(\left.
	2^{\sqrt{n}\sigma_n(Y_1^n)(\bar{S}_n - Q^{-1}(\epsilon))}
	\mathbb{I}_{\{\bar{S}_n \leq S_n\}}
	\right|X_1^n,Y_1^n\right) 
	> \frac{1}{\sqrt{n}}
	2^{-\sqrt{n}
	\bigl[T_n-(\sigma-\sigma_n(Y_1^n))Q^{-1}(\epsilon)\bigr] + C}\right].
	\label{eq:firstp}
\end{align}
For the conditional expectation, writing $\sigma_n$
for $\sigma_n(Y_1^n)$ for clarity, we have, 
\begin{align}
&
	\E\left.\left(2^{\sqrt{n}\sigma_n(\bar{S}_n-Q^{-1}(\epsilon))}
	\mathbb{I}_{\{\bar{S}_n \leq S_n\}}
	\right|X_1^n,Y_1^n\right) 
	\nonumber\\
&
	= \sum_{k=0}^{\infty}
	\E\left.\left(
	2^{\sqrt{n}\sigma_n(\bar{S}_n-Q^{-1}(\epsilon))}
	\mathbb{I}_{\{
	\sqrt{n}\sigma_nS_n-k-1
	<
	\sqrt{n}\sigma_n\bar{S}_n 
	\leq \sqrt{n}\sigma_nS_n-k
	\}}
	\right|X_1^n,Y_1^n\right) 
	\nonumber\\
&
	\leq 
	2^{\sqrt{n}\sigma_n(S_n-Q^{-1}(\epsilon))}
	\sum_{k=0}^{\infty}
	2^{-k}
	\BBP\left.\left[
	\sqrt{n}\sigma_nS_n-k-1
	<
	\sqrt{n}\sigma_n\bar{S}_n 
	\leq \sqrt{n}\sigma_nS_n-k
	\right|X_1^n,Y_1^n\right].
	\nonumber\\
&
	\leq 
	2^{\sqrt{n}\sigma_n(S_n-Q^{-1}(\epsilon))}
	\sum_{k=0}^{\infty}
	2^{-k}\Bigg\{
	\BBP\left.\left[
	\bar{S}_n 
	\leq 
	S_n-
	\frac{k}{\sqrt{n}\sigma_n}
	\right|X_1^n,Y_1^n\right]
	\label{eq:ink1}\\
&
	\hspace{2.2in}
	-\BBP\left.\left[
	\bar{S}_n 
	\leq 
	S_n-\frac{(k+1)}{\sqrt{n}\sigma_n}
	\right|X_1^n,Y_1^n\right]
	\Bigg\}.
	\label{eq:ink2}
\end{align}
Now following Lemma~47 of~\cite{PPV:10}, we note
that, conditional on $(X_1^n,Y_1^n)$, the only randomness
in the probabilities in~(\ref{eq:ink1}) and~(\ref{eq:ink2})
is in $\bar{S}_n$
via $\bar{X}_1^n$, so that we can apply the Berry-Ess\'een 
bound~\cite{fellerII:book}
twice to bound their difference,
resulting in,
\begin{align}
&
	2^{\sqrt{n}\sigma_n(S_n-Q^{-1}(\epsilon))}
	\sum_{k=0}^{\infty}
	2^{-k}
	\left[
	\Phi\left(S_n-\frac{k}{\sqrt{n}\sigma_n}\right)
	-\Phi\left(S_n-\frac{(k+1)}{\sqrt{n}\sigma_n}\right)
	+\frac{12m_3}{\sqrt{n}\sigma_n^3}
	\right]
	\nonumber\\
&
	\leq \frac{1}{\sqrt{n}}2^{\sqrt{n}\sigma_n(S_n 
	- Q^{-1}(\epsilon))}
	\Bigl( \frac{2}{\sqrt{2\pi}\sigma_n} 
	+ \frac{24m_3}{\sigma_n^3}\Bigr),
	\label{eq:taylor1}
\end{align}
where we have used a simple first-order Taylor 
expansion for $\Phi$, noting that $\phi(z)\leq 1/\sqrt{2\pi}$
for all $z$, and summed the geometric series.

Hence, combining all the estimates in~(\ref{eq:firstp})--(\ref{eq:taylor1}),
$$
	\BBP\left[\ell(f_n^*(X_1^n|Y_1^n)) > K_n\right]
\leq 
	\BBP\left[
	2^{
	\sqrt{n}\sigma(U_n
	-Q^{-1}(\epsilon))}
	\Bigl( \frac{2}{\sqrt{2\pi}\sigma_n} 
	+ \frac{24m_3}{\sigma_n^3}\Bigr)
	> 
	2^C \right],
$$
where we have defined,
$$U_n=
	\frac{1}{\sigma}[\sigma_n S_n+T_n] =
	 \frac{1}{\sigma\sqrt{n}}
	\sum_{i=1}^{n}{[-\log{P(X_i|Y_i) - H(X|Y)]}}.
$$
And we can further bound,
\begin{align}
&
	\BBP\left[\ell(f_n^*(X_1^n|Y_1^n)) > K_n\right]
	\nonumber\\
&\leq
	\BBP\left[
	2^{
	\sigma\sqrt{n}(U_n
	-Q^{-1}(\epsilon))}
	\Bigl( \frac{2}{\bar{v}^{1/2}} 
	+ \frac{24(2\pi)^{3/2}m_3}{\bar{v}^{3/2}}\Bigr)
	> 
	2^C \right]
+\BBP\left[\sigma^2_n(Y_1^n)<\frac{\bar{v}}{2\pi}\right].
\label{eq:twop}
\end{align}

For the first probability in~(\ref{eq:twop}) we have,
\begin{align*}
	\BBP\left[
	2^{
	\sigma\sqrt{n}(U_n
	-Q^{-1}(\epsilon))}
	2^{C-B}
	> 
	2^C \right]
&= 
	\BBP\left[U_n > Q^{-1}(\epsilon) 
	+ \frac{B}
	{\sigma\sqrt{n}}
	\right] \\
&\leq 
	Q\left(Q^{-1}(\epsilon) 
	+ \frac{B}
	{\sigma\sqrt{n}}
	\right) 
	+\frac{\E{[|-\log{P(X|Y)}-H(X|Y)|^3]}}{2\sigma^3\sqrt{n}},
\end{align*}
where we used the Berry-Ess\'{e}en bound~\cite{petrov-book:95}
for the normalised partial sum $U_n$ of the i.i.d.\ random variables 
$\{-\log{P(X_i|Y_i)}\}$ 
with mean $H(X|Y)$ and variance $\sigma^2$. 
And a second-order Taylor expansion of $Q$,
using the fact that,
$0 \leq Q''(x) = x\phi(x) \leq \frac{1}{\sqrt{2\pi e}},$
for all $x \geq 0$, gives,
\begin{align}
	\BBP\Big[
	2^{
	\sigma\sqrt{n}(U_n
	-Q^{-1}(\epsilon))}
&
	\;
	2^{C-B}
	\; 
	>2^C \Big]
	\nonumber\\
& 
	\leq
	\epsilon
	-B\phi(Q^{-1}(\epsilon))
	+\frac{B^2}{2\sqrt{2\pi e}}
	+\frac{\E{[|-\log{P(X|Y)}-H(X|Y)|^3]}}{2\sigma^3\sqrt{n}}.
\label{eq:midBE}
\end{align}
For the second probability in~(\ref{eq:twop}), a simple
application of Chebyshev's inequality gives,
\be
\BBP\left[\frac{1}{n}\sum_{i=1}^nV(Y_i)<\frac{\bar{v}}{2\pi}\right]
\leq\frac{\psi^2}{n(1-\frac{1}{2\pi})^2\bar{v}^2}.
\label{eq:cheb}
\ee
After substituting the bounds~(\ref{eq:midBE}) and~(\ref{eq:cheb})
in~(\ref{eq:twop}), simple algebra shows that,
for all $n$ satisfying~(\ref{nthrescondrateub}),
the probability is $\leq\epsilon$, completing the proof.
\end{proof}

\medskip

Next we prove a corresponding converse bound.
Once again we observe that, 
by the definitions in Section~\ref{descriptionofanoptimal},
Theorem~\ref{IIDAnnealedConverse}
also holds for $R_p^*(n,\epsilon)$ in the case
of prefix-free codes.

\begin{theorem}[Pair-based converse]
\label{IIDAnnealedConverse}
Let $(\Xp,\Yp)$ be an i.i.d.\ source-side
information pair, with conditional varentropy rate
$\sigma^2=\sigma^2(X|Y)>0$.
For any $0 < \epsilon < \frac{1}{2},$
the pair-based optimal compression rate satisfies,
\begin{equation*}
R^*(n,\epsilon) \geq H(X|Y) 
+ \frac{\sigma(X|Y)}{\sqrt{n}} Q^{-1}(\epsilon) -\frac{\log{n}}{2n} 
- \frac{C'}{n},
\end{equation*}
for all,
\begin{equation}\label{AnConvLBThres}
n > \frac{C'^2}{4(Q^{-1}(\epsilon))^2\sigma^2},
\end{equation} 
where,
$$
C' = \frac{\E{[|-\log{P(X|Y)}-H(X|Y)|^3]} + 2\sigma^3}
{2\sigma^2\phi(Q^{-1}(\epsilon))}.
$$
\end{theorem}

\begin{proof}
Using the Berry-Ess\'{e}en bound~\cite{fellerII:book}
we have,
\begin{align*}
&\PP\biggl[ -\log{P(X_{1}^{n}|Y_{1}^{n})} 
\geq nH(X|Y) + \sqrt{n}\sigma Q^{-1}(\epsilon) - C'\biggr] \\
&= \PP\biggl[ \frac{1}{\sigma\sqrt{n}}
\sum_{i=1}^{n}{[-\log{P(X_{i}|Y_{i})} - H(X|Y)]}
\geq Q^{-1}(\epsilon) - \frac{C'}{\sigma\sqrt{n}}\biggr] \\ 
&\geq 
 Q\left(Q^{-1}(\epsilon)-\frac{C'}{\sigma\sqrt{n}}\right) 
- \frac{\E{|-\log{P(X|Y)}-H(X|Y)|^3}}{2\sigma^3\sqrt{n}},
\end{align*}
and using the simple earlier bound (\ref{QineqConv}), 
noting that~(\ref{AnConvLBThres}) implies that
the condition in (\ref{QineqConv}) is satisfied,
\begin{align*}
&\PP\biggl[ -\log{P(X_{1}^{n}|Y_{1}^{n})} 
\geq nH(X|Y) + \sqrt{n}\sigma Q^{-1}(\epsilon) - C'\biggr] \\
&\geq  \epsilon + \frac{C'}{\sigma\sqrt{n}}
\phi(Q^{-1}(\epsilon)) 
- \frac{\E{|-\log{P(X|Y)}-H(X|Y)|^3}}{2\sigma^3\sqrt{n}} \\ 
&= \epsilon + \frac{1}{\sqrt{n}},
\end{align*}
Now applying the general converse result
in Theorem \ref{generalconvSI} 
with $\tau=\tau_n = \frac{1}{2}\log{n}$ and $X_{1}^{n}$ in place
of $X$,
\begin{align*}
&\PP\Bigl[ \ell(f_n^{*}(X_{1}^{n}|Y_{1}^{n})) \geq  
nH(X|Y) + \sqrt{n}\sigma Q^{-1}(\epsilon) - C' -\frac{\log{n}}{2} \Bigr] \\
&\geq \PP\biggl[ -\log{P(X_{1}^{n}|Y_{1}^{n})} \geq nH(X|Y) + \sqrt{n}\sigma Q^{-1}(\epsilon) - C' \biggr] - \frac{1}{\sqrt{n}}
\geq \epsilon,
\end{align*}
and the claimed bound follows.
\end{proof}

\section{Normal Approximation for Markov Sources}
\label{normalmarkovsection}

In this section we consider extensions of the normal 
approximation bounds for the optimal rate
in Sections~\ref{memorylessnormalsection}
and~\ref{annealedsectionmemoryless},
to the case of Markov sources. Note that 
the results of Section~\ref{memorylessnormalsection}
for the reference-based optimal rate $R^*(n,\epsilon|y_1^n)$
apply not only to the case of i.i.d.\ source-side
information pairs $(\Xp,\Yp)$, but much more generally
to arbitrary side-information sources $\Yp$ as long
as $\Xp$ is conditionally i.i.d.\ given $\Yp$. This 
is a broad class including, 
among others, all hidden Markov models $\Xp$. For this
reason, we restrict our attention here
to the pair-based optimal rate $R^*(n,\epsilon)$.

As discussed in the context of  compression
without side information \cite{kontoyiannis-verdu:14},
the Berry-Ess\'{e}en bound for Markov chains
is not known to hold at the
same level of generality as in the i.i.d.\ case. In fact,
even for restricted class of reversible chains 
where an explicit Berry-Ess\'{e}en bound is known
\cite{mann:phd},
it involves constants that are larger than 
those in the i.i.d.\ case by more than
four orders of magnitude, making any resulting
bounds significantly
less relevant in practice.

Therefore, in the Markov case we employ a general result
of Nagaev \cite{nagaev:61} that does not lead to explicit 
values for the relevant constants, but which applies to
all ergodic Markov chains. For similar reasons, 
rather than attempting to generalise the rather involved
proof of the achievability result in Theorem~\ref{IIDAnnealedAchievability},
we choose to illustrate a much simpler argument that
leads to a weaker bound, not containing the
third-order $(\log n)/2n$ term as in~(\ref{eq:iidachieve}).

\medskip

Throughout this section we consider
a source-side information pair
$(\Xp,\Yp)$ which is an irreducible and aperiodic, 
$d$th order Markov
chain, with 
conditional varentropy rate
$\sigma^2(\Xp|\Yp)$ as in Lemma~\ref{lem:varentropy},
and we also assume that the side information 
process $\Yp$ itself is an irreducible and aperiodic, 
$d$th order Markov chain.
Note that we allow $(\Xp,\Yp)$ to have an
arbitrary initial distribution, so that in 
particular we do not assume it is stationary.
The main probabilistic tool we will need in the proof
of Theorem~\ref{Markovboth}
is the following normal approximation
bound for the conditional information density
$-\log P(X_1^n|Y_1^n)$;
it is proved in Appendix~\ref{s:bemarkovlogp}.

\begin{theorem}
[Berry-Ess\'een bound for the conditional information density]
\label{bemarkovlogp}
Suppose the source-side information pair $(\Xp,\Yp)$ 
and the side information process $\Yp$ are
$d$th order, irreducible and aperiodic Markov chains,
with conditional entropy rate $H=H(\Xp|\Yp)$ and conditional varentropy rate
$\sigma^2=\sigma^2(\Xp|\Yp)>0.$
Then
there exists a finite constant $A > 0$ such that, for all $n\geq 1$,
\begin{equation*}
\sup_{z \in \mathbb{R}}
\left|\BBP\left[
\frac{-\log{P(X_1^n|Y_1^n)} -nH}
{\sigma\sqrt{n}}>z
\right] 
 - Q(z)\right|
\leq \frac{A}{\sqrt{n}},
\end{equation*} 
\end{theorem}

\begin{theorem}
[Normal approximation for Markov sources]
\label{Markovboth}
Suppose the source-side information pair $(\Xp,\Yp)$ 
and the side information process $\Yp$ are
$d$th order, irreducible and aperiodic Markov chains,
with conditional entropy rate $H=H(\Xp|\Yp)$ and conditional varentropy rate
$\sigma^2=\sigma^2(\Xp|\Yp)>0.$
Then, for any $\epsilon\in(0,1/2)$,
there are finite constants $C_m,C_m'$ and integers $N,N'$
such that,
\begin{equation}
R^*(n,\epsilon) \leq H(\Xp|\Yp) 
+ \frac{\sigma(\Xp|\Yp)}{\sqrt{n}}Q^{-1}(\epsilon) 
+ \frac{C_m}{n},
\qquad\mbox{for all}\;n\geq N,
\label{eq:markovU}
\end{equation}
and,
\begin{equation}
R^*(n,\epsilon) \geq H(\Xp|\Yp) 
+ \frac{\sigma(\Xp|\Yp)}{\sqrt{n}}Q^{-1}(\epsilon) 
- \frac{\log{n}}{2n} - \frac{C_m'}{n},
\qquad\mbox{for all}\;n\geq N'.
\label{eq:markovL}
\end{equation}
\end{theorem}

Note that the only reason we do not give explicit
values for the constants $C_m,C_m',N,N'$ 
is because of the unspecified
constant in the Berry-Ess\'{e}en bound. In fact,
for any class of Markov chains for which the
constant $A$ in Theorem~\ref{bemarkovlogp}
is known explicitly, we can take,
$$C_m = \frac{2A\sigma}{\phi(Q^{-1}(\epsilon))},\qquad
C_m' = \frac{\sigma(A+1)}{\phi(Q^{-1}(\epsilon))},$$
and,
$$N= \frac{2A^2}{\pi e (\phi(Q^{-1}(\epsilon)))^4},\qquad
N'=
\Bigl(\frac{A+1}{Q^{-1}(\epsilon)\phi(Q^{-1}(\epsilon))}\Bigr)^2.
$$

As with the corresponding results for memoryless 
sources, Theorems~\ref{IIDAnnealedAchievability}
and~\ref{IIDAnnealedConverse},
we observe
that both~(\ref{eq:markovU}) 
and~(\ref{eq:markovL}) in 
Theorem~\ref{Markovboth}
remain valid for $R_p^*(n,\epsilon)$ in 
the case of prefix-free codes.

\begin{proof}
Let $A$ be the constant of Theorem~\ref{bemarkovlogp}.
Taking $C_m = 2A\sigma/\phi(Q^{-1}(\epsilon))$ and 
$K_n = nH + \sigma\sqrt{n}Q^{-1}(\epsilon) + C_m$,
the general achievability bound in Theorem~\ref{t:achieve}
gives,
\begin{align*}
\Pbig{\ell(f_n^*(X_1^n|Y_1^n)) \geq K_n} 
&\leq \Pbig{-\log{P(X_1^n|Y_1^n)} \geq K_n} \\
&\leq Q\Big(Q^{-1}(\epsilon) + \frac{C_m}{\sigma\sqrt{n}}\Big) 
+ \frac{A}{\sqrt{n}},
\end{align*}
where the second inequality follows from Theorem~\ref{bemarkovlogp}.
Since $Q''(x) \leq \frac{1}{\sqrt{2\pi e}}, x \geq 0$, 
a second-order Taylor expansion for $Q$ yields,
\begin{equation*}
\Pbig{\ell(f_n^*(X_1^n|Y_1^n)) \geq K_n} 
\leq \epsilon - \frac{C_m}{\sigma\sqrt{n}}\Bigl\{\phi(Q^{-1}(\epsilon)) 
- \frac{C_m}{2 \sigma\sqrt{2\pi e n}} - \frac{A\sigma}{C_m} \Bigr\} \leq \epsilon,
\end{equation*}
where the last inequality holds for all,
$n \geq 2A^2/[\pi e (\phi(Q^{-1}(\epsilon)))^4]$.
This proves~(\ref{eq:markovU}).

For the converse, taking
$C_m' = \sigma(A+1)/\phi(Q^{-1}(\epsilon))$ and 
$K_n = nH + \sigma\sqrt{n}Q^{-1}(\epsilon) - (\log{n})/2 - C_m'$,
and $\tau =(\log{n})/2$,
the general converse bound in 
Theorem~\ref{generalconvSI} gives,
\ben
\Pbig{\ell(f_n^*(X_1^n|Y_1^n)) \geq K_n} 
&\geq&
	\BBP\left[\frac{-\log P(X_1^n|Y_1^n)-nH}{\sigma\sqrt{n}}
	\geq Q^{-1}(\epsilon) - \frac{C_m'}{\sigma\sqrt{n}}\right]
	-\frac{1}{\sqrt{n}}
	\\
&\geq&
Q\Bigl(Q^{-1}(\epsilon) - \frac{A+1}{\phi(Q^{-1}(\epsilon))\sqrt{n}}\Bigr) 
- \frac{A+1}{\sqrt{n}},
\een
where the second bound follows from
Theorem~\ref{bemarkovlogp}.
Finally, using a simple first-order Taylor expansion
of $Q$, and noting that
$\phi(x)$ is nonincreasing for $x \geq 0$
and that,
$$Q^{-1}(\epsilon) - \frac{A+1}{\phi(Q^{-1}(\epsilon))\sqrt{n}} \geq 0,$$
for $\epsilon \in (0,1/2)$
and,
\begin{equation*}
n \geq \Bigl(\frac{A+1}{Q^{-1}(\epsilon)\phi(Q^{-1}(\epsilon))}\Bigr)^2,
\end{equation*}
yields that
$\Pbig{\ell(f_n^*(X_1^n|Y_1^n)) \geq K_n}>\epsilon$.
This gives~(\ref{eq:markovL}) and completes the proof.
\end{proof} 

\begin{appendices}

\appendixpage

\section{Proof of Lemma~\ref{lem:varentropy}. } 
\label{app:Mvar}
For any pair of strings $(x_1^{n+d},y_1^{n+d}) \in 
(\mathcal{X}\times\mathcal{Y})^{n+d}$ such that 
$P(x_1^{n+d},y_1^{n+d}) > 0$,
we have,
\begin{align}
-\log{P(x_1^n|y_1^n)} 
=&\; 
	\log\left(\frac{P(y_1^d)\prod_{j=d+1}^n{P(y_j|y_{j-d}^{j-1})}}
	{P(x_1^d,y_1^d)\prod_{j=d+1}^n
	{P(x_j,y_j|x_{j-d}^{j-1},y_{j-d}^{j-1})}}\right) 
	\nonumber\\ 
=&\;
	\sum_{j=d+1}^{d+n}{
	\log\left(
	\frac{P(y_j|y_{j-d}^{j-1})}{P(x_j,y_j|x_{j-d}^{j-1},y_{j-d}^{j-1})}
	\right)} 
	\nonumber\\
&\;\;
	- \log\left(\frac{P(x_1^d,y_1^d)
	\prod_{j=n+1}^{n+d}{P(x_j,y_j| x_{j-d}^{j-1},y_{j-d}^{j-1})}}
	{P(y_1^d)\prod_{j=n+1}^{n+d}{P(y_j|y_{j-d}^{j-1})}}\right)
	\nonumber\\ 
\label{logpfplusdelta}
&= \sum_{j=1}^n{f(x_j^{j+d},y_j^{j+d})} + \Delta_n,
\end{align}
where $f:(\mathcal{X}\times\mathcal{Y})^{d+1} \rightarrow \mathbb{R}$ 
is defined on,
\begin{equation*}
\mathcal{S} = \Bigl\{(x_1^{d+1},y_1^{d+1}) 
	\in (\mathcal{X}\times\mathcal{Y})^{d+1}: 
P\bigl(x_{d+1},y_{d+1}|x_1^d,y_1^{d}\bigr) > 0\Bigr\},
\end{equation*}
by,
\begin{equation*}
f(x_1^{d+1},y_1^{d+1}) 
= \log\left(\frac{P(y_{d+1}|y_{1}^{d})}{P(x_{d+1},y_{d+1}
|x_1^d,y_{1}^{d})}\right),
\end{equation*}
and,
\begin{equation*}
\Delta_n = 
\log\left(
\frac{P(y_1^d)
\prod_{j=n+1}^{n+d}P(y_j|y_{j-d}^{j-1})}
{P(x_1^d,y_1^d)\prod_{j=n+1}^{n+d}
{P(x_j,y_j|x_{j-d}^{j-1},y_{j-d}^{j-1})}}\right).
\end{equation*}
Taking the maximum of $\Delta_n$ over 
all nonzero-probability strings, gives a maximum 
of finitely many terms all of which are finite, 
so,
\begin{equation} \label{supdeltabounded}
\delta = \max{|\Delta_n|} < \infty.
\end{equation}

Let $\Zp=\seq{Z_n}$ denote the first-order Markov chain defined 
by taking overlapping $(d+1)$-blocks in the joint process,
\begin{equation*}
Z_n = ((X,Y)_n,(X,Y)_{n+1},\ldots,(X,Y)_{n+d}).
\end{equation*}
Since $(\Xp,\Yp)$ is irreducible and aperiodic, 
so is $\Zp$, so it has a unique stationary distribution
$\pi$.
Let $\seq{\tilde{Z}_n}$ denote a stationary version of $\seq{Z_n}$, 
with the same transition probabilities as $\seq{Z_n}$
and with $\tilde{Z}_1^d\sim\pi$.
And using~(\ref{logpfplusdelta})
we can express,
$$\frac{1}{n}H(X_1^n|Y_1^n)=\frac{1}{n}E[-\log P(X_1^n|Y_1^n)]
=E\left[
\frac{1}{n}\sum_{j=1}^nf(X_j^{j+d},Y_j^{j+d})
\right]
+\frac{\bar{\Delta}_n}{n},$$
where the $\bar{\Delta}_n=E(\Delta_n)$ 
are constants all absolutely bounded by $\delta<\infty$.

Then the $L^1$ ergodic theorem for Markov chains,
see, e.g., \cite[p.~88]{chung:book}, implies
that the limit,
$$H(\Xp|\Yp)=\lim_{n\to\infty} \frac{1}{n}H(X_1^n|Y_1^n),$$
exists and it equals $E[f(\tilde{Z}_1)]$,
independently of the initial distribution
of $(\Xp,\Yp)$. Similarly, we can write the variances,
\be
\frac{1}{n}\VAR(-\log P(X_1^n|Y_1^n))
&=&
	\frac{1}{n}\VAR\left(
	\sum_{j=1}^nf(X_j^{j+d},Y_j^{j+d})+\Delta_n
	\right)
	\nonumber\\
&=&
	\frac{1}{n}\E\left\{\left[
	\sum_{j=1}^nf(X_j^{j+d},Y_j^{j+d})
	-\E\left(\sum_{j=1}^nf(X_j^{j+d},Y_j^{j+d})\right)
	\right]^2
	\right\}+o(1)
	\nonumber\\
&=&
	\frac{1}{n}\E\left\{\left[
	\sum_{j=1}^nf(X_j^{j+d},Y_j^{j+d})
	-\E\left(\sum_{j=1}^nf(Z_j)\right)
	\right]^2
	\right\}
	\label{eq:L2a}\\
&&
	-
	\frac{1}{n}\left[
	\E\left(\sum_{j=1}^nf(X_j^{j+d},Y_j^{j+d})
	-f(Z_j)\right)
	\right]^2
	+o(1),
	\label{eq:L2b}
\ee
where the first step uses the
uniform boundedness of $\Delta_n$.
Then, the $L^2$ ergodic theorem
in \cite[p.~97]{chung:book}, implies
that, independently of the initial distribution
of $(\Xp,\Yp)$, the limit of the term
in~(\ref{eq:L2a}) exists and equals $\sigma^2(\Xp|\Yp)$,
and the limit of the terms in~(\ref{eq:L2b}) is zero.
\qed

\section{Proof of Theorem~\ref{bemarkovlogp}}
\label{s:bemarkovlogp}

We adopt the same setting and notation as in the proof 
of Lemma~\ref{lem:varentropy} above.

Since the function
$f$ is bounded we can apply \cite[Theorem~1]{nagaev:61} to obtain
that there exists a finite constant $A_1$ such that, for all $n \geq 1$,
\begin{equation}
\label{eq:nagaev}
\sup_{z \in \mathbb{R}}
\left|\BBP\left[
\frac{\sum_{j=1}^n{f(X_j,Y_j)} -nH}{\sigma\sqrt{n}}>z
\right] - Q(z)\right|
\leq \frac{A_1}{\sqrt{n}},
\end{equation} 
where $H$ can also be expressed
as $H= \mathbb{E}[f(\tilde{Z}_1)]$, and
since the function $f$ is bounded and the distribution of the 
chain $\seq{Z_n}$ converges to the stationary distribution 
exponentially fast, the conditional varentropy is also given by,
\begin{equation*} 
\sigma^2 = \lim_{n \rightarrow \infty}
{\frac{1}{n}\mathbb{E}\Bigl(\sum_{j=1}^n{[f(\tilde{Z}_j) - H]}\Bigr)^2},
\end{equation*}
and it coincides with the expression 
in Lemma~\ref{lem:varentropy}.

For $z\in\RL$, define,
\begin{align*}
F_n(z) &= \BBP\left[\frac{-\log{P(X_1^n|Y_1^n)} -nH}
	{\sigma\sqrt{n}}>z\right] \\
G_n(z) &= \BBP\left[\frac{\sum_{j=1}^n{f(X_j,Y_j)} -nH}
	{\sigma\sqrt{n}}>z\right].
\end{align*}
Since $F_n$ and $G_n$ are non-increasing,~(\ref{logpfplusdelta}),
(\ref{supdeltabounded}), and~(\ref{eq:nagaev}) yield,
\begin{align*}
F_n(z) \geq G_n\Big(z + \frac{\delta}{\sigma\sqrt{n}}\Big) 
\geq Q\Big(z + \frac{\delta}{\sigma\sqrt{n}}\Big) - \frac{A_1}{\sqrt{n}} 
\geq Q(z) - \frac{A}{\sqrt{n}},
\end{align*}
uniformly in $z$, where in the last inequality we used a simple
first-order Taylor expansion for $Q$, with 
$A = A_1 + \frac{\delta}{\sqrt{2\pi}}$.
Similarly, we have,
$$F_n(z) \leq G_n\Big(z - \frac{\delta}{\sigma\sqrt{n}}\Big) 
\leq Q\Big(z - \frac{\delta}{\sigma\sqrt{n}}\Big) + \frac{A_1}{\sqrt{n}} 
\leq Q(z) + \frac{A}{\sqrt{n}},
$$
uniformly in $z$.
\qed

\end{appendices}

\renewcommand{\baselinestretch}{0.9}

{\small

\bibliographystyle{plain}

\def\cprime{$'$}

}


\begin{thebibliography}{10}

\bibitem{aaron:02b}
A.~Aaron and B.~Girod.
\newblock Compression with side information using turbo codes.
\newblock In {\em 2002 Data Compression Conference}, pages 252--261, Snowbird,
  UT, April 2002.

\bibitem{aaron:02}
A.~Aaron, R.~Zhang, and B.~Girod.
\newblock {Wyner-Ziv} coding of motion video.
\newblock In {\em 36th Asilomar Conference on Signals, Systems and Computers},
  volume~1, pages 240--244, Pacific Grove, CA, November 2002.

\bibitem{C-K-Verdu:06}
H.~Cai, S.R. Kulkarni, and S.~Verd{\'u}.
\newblock An algorithm for universal lossless compression with side
  information.
\newblock {\em IEEE Trans. Inform. Theory}, 52(9):4008--4016, September 2006.

\bibitem{chen:14}
C.L.P. Chen and C.Y. Zhang.
\newblock Data-intensive applications, challenges, techniques and technologies:
  {A} survey on big data.
\newblock {\em Information Sciences}, 275:314--347, August 2014.


\bibitem{kostina:19}
S.~{Chen}, M.~{Effros}, and V.~{Kostina}, ``Lossless source coding in the
  point-to-point, multiple access, and random access scenarios,'' \emph{IEEE
  Transactions on Information Theory}, vol.~66, no.~11, pp. 6688--6722.
  
\bibitem{chung:book}
K.L. Chung.
\newblock {\em Markov chains with stationary transition probabilities}.
\newblock Springer-Verlag, New York, 1967.


  
\bibitem{cover:book2}
T.M. Cover and J.A. Thomas.
\newblock {\em Elements of information theory}.
\newblock J. Wiley \& Sons, New York, second edition, 2012.

\bibitem{fellerII:book}
W.~Feller.
\newblock {\em An introduction to probability theory and its applications.
  Vol.~II}.
\newblock John Wiley \& Sons Inc., New York, second edition, 1971.

\bibitem{fritz:11}
M.H.Y. Fritz, R.~Leinonen, G.~Cochrane, and E.~Birney.
\newblock Efficient storage of high throughput {DNA} sequencing data using
  reference-based compression.
\newblock {\em Genome Research}, 21(5):734--740, 2011.

\bibitem{irikosut:15}
N.~{Iri} and O.~{Kosut}.
\newblock Third-order coding rate for universal compression of {Markov}
  sources.
\newblock In {\em 2015 IEEE International Symposium on Information Theory
  (ISIT)}, pages 1996--2000, Hong Kong, June 2015.

\bibitem{jacob:08}
T.~Jacob and R.K. Bansal.
\newblock On the optimality of sliding window {Lempel-Ziv} algorithm with side
  information.
\newblock In {\em 2008 International Symposium on Information Theory and its
  Applications (ISITA)}, pages 1--6, December 2008.

\bibitem{jain-bansal}
A.~Jain and R.K. Bansal.
\newblock On optimality and redundancy of side information version of {SWLZ}.
\newblock In {\em 2017 IEEE International Symposium on Information Theory
  (ISIT)}, pages 306--310, Aachen, Germany, June 2017.

\bibitem{jose:18}
S.T. Jose and A.A. Kulkarni.
\newblock Improved finite blocklength converses for {Slepian-Wolf} coding via
  linear programming.
\newblock {\em IEEE Trans. Inform. Theory}, 65(4):2423--2441, April 2019.

\bibitem{kontoyiannis-97}
I.~Kontoyiannis.
\newblock Second-order noiseless source coding theorems.
\newblock {\em IEEE Trans. Inform. Theory}, 43(4):1339--1341, July 1997.

\bibitem{kontoyiannis-verdu:14}
I.~Kontoyiannis and S.~Verd\'{u}.
\newblock Optimal lossless data compression: {N}on-asymptotics and asymptotics.
\newblock {\em IEEE Trans. Inform. Theory}, 60(2):777--795, February 2014.

\bibitem{mann:phd}
B.~Mann.
\newblock {\em Berry-Ess\'{e}en central limit theorem for Markov chains}.
\newblock PhD thesis, Department of Mathematics, Harvard University, Cambridge,
  MA, 1996.

\bibitem{miyake:95}
S.~Miyake and F.~Kanaya.
\newblock Coding theorems on correlated general sources.
\newblock {\em IEICE Trans. Fundam. Electron. Commun. Comput. Sci.},
  78(9):1063--1070, 1995.

\bibitem{nagaev:61}
S.V. Nagaev.
\newblock More exact limit theorems for homogeneous {M}arkov chains.
\newblock {\em Theory Probab. Appl.}, 6(1):62--81, 1961.

\bibitem{nomura:14}
R.~{Nomura} and T.S. {Han}.
\newblock Second-order {Slepian-Wolf} coding theorems for non-mixed and mixed
  sources.
\newblock {\em IEEE Trans. Inform. Theory}, 60(9):5553--5572, September 2014.

\bibitem{nomura:11}
R.~Nomura and T.~Matsushima.
\newblock On the overflow probability of fixed-to-variable length codes with
  side information.
\newblock {\em IEICE Trans. Fundam. Electron. Commun. Comput. Sci.},
  94(11):2083--2091, November 2011.

\bibitem{pavlichin:18}
D.~Pavlichin, T.~Weissman, and G.~Mably.
\newblock The quest to save genomics: {Unless} researchers solve the looming
  data compression problem, biomedical science could stagnate.
\newblock {\em IEEE Spectrum}, 55(9):27--31, September 2018.

\bibitem{petrov-book:95}
V.V. Petrov.
\newblock {\em Limit theorems of probability theory}.
\newblock The Clarendon Press, Oxford University Press, New York, 1995.

\bibitem{PPV:10}
Y.~Polyanskiy, H.V. Poor, and S.~Verd{\'u}.
\newblock Channel coding rate in the finite blocklength regime.
\newblock {\em IEEE Trans. Inform. Theory}, 56(5):2307--2359, May 2010.

\bibitem{pradhan:01}
S.S. Pradhan and K.~Ramchandran.
\newblock Enhancing analog image transmission systems using digital side
  information: {A} new wavelet-based image coding paradigm.
\newblock In {\em 2001 Data Compression Conference}, pages 63--72, Snowbird,
  UT, March 2001.


\bibitem{sakai-arxiv:19}
Y.~{Sakai} and V.~Y.~F. {Tan}, ``Variable-length source dispersions differ
  under maximum and average error criteria,'' \emph{IEEE Transactions on
  Information Theory}, vol.~66, no.~12, pp. 7565--7587, 2020.

\bibitem{slepianwolf:73}
D.~{Slepian} and J.~{Wolf}.
\newblock Noiseless coding of correlated information sources.
\newblock {\em IEEE Trans. Inform. Theory}, 19(4):471--480, July 1973.

\bibitem{stites:00}
R.~Stites and J.~Kieffer.
\newblock Resolution scalable lossless progressive image coding via conditional
  quadrisection.
\newblock In {\em 2000 International Conference on Image Processing}, volume~1,
  pages 976--979, Vancouver, BC, September 2000.

\bibitem{strassen:64b}
V.~Strassen.
\newblock Asymptotische {A}bsch\"atzungen in {S}hannons {I}nformationstheorie.
\newblock In {\em Trans. Third Prague Conf. Information Theory, Statist.
  Decision Functions, Random Processes (Liblice, 1962)}, pages 689--723. Publ.
  House Czech. Acad. Sci., Prague, 1964.

\bibitem{subrahmanya:95}
P.~{Subrahmanya} and T.~{Berger}.
\newblock A sliding window {Lempel-Ziv} algorithm for differential layer
  encoding in progressive transmission.
\newblock In {\em 1995 IEEE International Symposium on Information Theory
  (ISIT)}, page 266, Whistler, BC, September 1995.

\bibitem{suel:02}
T.~Suel and N.~Memon.
\newblock Algorithms for delta compression and remote file synchronization.
\newblock In K.~Sayood, editor, {\em Lossless Compression Handbook}. Academic
  Press, 2002.


\bibitem{tan:12}
V.~Y.~F. {Tan} and O.~{Kosut}, ``On the dispersions of three network
  information theory problems,'' \emph{IEEE Transactions on Information
  Theory}, vol.~60, no.~2, pp. 881--903, 2014.
  
\bibitem{tock:05}
T.~Tock and Y.~Steinberg.
\newblock On conditional entropy and conditional recurrence time.
\newblock Unpublished manuscript, February, 2005.

\bibitem{rsync}
A.~Tridgell and P.~Mackerras.
\newblock The rsync algorithm.
\newblock Technical report TR-CS-96-05, The Australian National University,
  Canberra, Australia, June 1996.

\bibitem{uyematsu:03}
T.~Uyematsu and S.~Kuzuoka.
\newblock Conditional {Lempel-Ziv} complexity and its application to source
  coding theorem with side information.
\newblock {\em IEICE Trans. Fundam. Electron. Commun. Comput. Sci.},
  E86-A(10):2615--2617, October 2003.

\bibitem{yang:01}
E.-H. Yang, A.~Kaltchenko, and J.C. Kieffer.
\newblock Universal lossless data compression with side information by using a
  conditional {MPM} grammar transform.
\newblock {\em IEEE Trans. Inform. Theory}, 47(6):2130--2150, September 2001.

\end{thebibliography}

\end{document}